\documentclass{amsart}
\usepackage{amsmath}
\usepackage{amsfonts}
\usepackage{amsthm, upref}
\usepackage{graphicx}
\usepackage{enumerate}
\usepackage[usenames, dvipsnames]{color}
\usepackage{url}
\usepackage{tikz}
\usetikzlibrary{decorations.pathreplacing}
\input xy
\xyoption{all}

\newcommand{\comment}[1]{}
\newcommand{\eq}{\begin{equation}}
\newcommand{\en}{\end{equation}}
\newcommand{\rr}{\mathbb{R}}

\newcommand{\norm}[1]{\left\lVert #1 \right\rVert}

\newcommand{\mcal}[1]{\mathcal{#1}}
\newcommand{\iprod}[1]{\left\langle #1 \right\rangle }

\newcommand{\nin}{\noindent}
\newcommand{\tbf}{\textbf}
\newcommand{\QQ}{\mathcal{Q}}
\newcommand{\RR}{\mathcal{R}}
\newcommand{\PP}{\mathcal{P}}
\newcommand{\PPT}{\widetilde{\PP}}

\newcommand{\exQ}{\widetilde{\mathcal{Q}}}
\newcommand{\exP}{\widetilde{\mathcal{P}}}
\newcommand{\simp}{\Delta}

\newcommand{\diverge}{T}
\newcommand{\Hess}{{\mathrm{Hess}}}

\newcommand{\E}{\mathrm{E}}

\begin{document}

\theoremstyle{plain}
\newtheorem{thm}{Theorem}
\newtheorem{lemma}[thm]{Lemma}
\newtheorem{prop}[thm]{Proposition}
\newtheorem{cor}[thm]{Corollary}

\theoremstyle{definition}
\newtheorem{defn}{Definition}
\newtheorem{asmp}{Assumption}
\newtheorem{notn}{Notation}
\newtheorem{prb}{Problem}

\theoremstyle{remark}
\newtheorem{rmk}{Remark}
\newtheorem{exm}{Example}
\newtheorem{clm}{Claim}


\title[The geometry of relative arbitrage]{The Geometry of relative arbitrage}
\author[S. Pal and T.-K. L. Wong]{Soumik Pal and Ting-Kam Leonard Wong}
\address{Department of Mathematics\\ University of Washington\\ Seattle, WA 98195}
\email{soumikpal@gmail.com, tkleonardwong@gmail.edu}
\keywords{Stochastic portfolio theory, rebalancing, functionally generated portfolios, optimal transport, model-free finance}
\subjclass[2000]{Primary 00A30; Secondary 00A22, 03E20}
\thanks{This research is partially supported by NSF grant DMS-1308340.}
\date{\today}

\begin{abstract}
Consider an equity market with $n$ stocks. The vector of proportions of the total market capitalizations that belong to each stock is called the market weight. The market weight defines the market portfolio which is a buy-and-hold portfolio representing the performance of the entire stock market. Consider a function that assigns a portfolio vector to each possible value of the market weight, and we perform self-financing trading using this portfolio function. We study the problem of characterizing functions such that the resulting portfolio will outperform the market portfolio in the long run under the conditions of diversity and sufficient volatility. No other assumption on the future behavior of stock prices is made. We prove that the only solutions are functionally generated portfolios in the sense of Fernholz. A second characterization is given as the optimal maps of a remarkable optimal transport problem. Both characterizations follow from a novel property of portfolios called multiplicative cyclical monotonicity.
\end{abstract}

\maketitle

\section{Introduction}
Consider investing in an equity market. At each point in time the investor allocates the current wealth among the stocks and form a portfolio. We will only consider self-financing portfolio strategies that are fully invested in the stock market, are long-only, and are never allowed to borrow or lend in the money market. Mathematically, consider the closed unit simplex in ${\Bbb R}^n$ defined by 
\eq\label{eq:unitsimplex}
\overline{\simp^{(n)}}=\left\{   p=\left(p_1, \ldots , p_n  \right) \in {\Bbb R}^n,\; p_i\ge 0\; \text{for all $i$}, \; \text{and}\; \sum_{i=1}^n p_i=1  \right\}.
\en
A portfolio at any point of time is represented by a vector in $\overline{\simp^{(n)}}$ where $n$ is the number of stocks. The individual coordinates of this vector represent the proportions of wealth invested in each stock and are called portfolio weights. Over time this leads to a process with state space $\overline{\simp^{(n)}}$.

For example, a {\it buy-and-hold portfolio} is one where one buys a certain number of shares of each stock initially and holds them for all future time. Of special importance is the {\it market portfolio} defined as follows. Let $X_i(t) > 0$ be the market capitalization of stock $i$ at time $t$. The market weight of stock $i$ is defined by
\begin{equation} \label{eqn:marketweight}
\mu_i(t) = \frac{X_i(t)}{X_1(t) + \cdots + X_n(t)}, \quad i = 1, \ldots, n,
\end{equation}
and the market portfolio is the portfolio with weights $\mu(t) = (\mu_1(t),\ldots , \mu_n(t))$. The market weight takes values in the open unit simplex $\Delta^{(n)}$. The value of this portfolio reflects the growth of the entire equity market and is called a market index. For example, in the US equity market, a standard benchmark is the S\&P500 index which is approximately the value process of a buy-and-hold portfolio. 

Due to the special importance of the market portfolio as an investment benchmark and as an efficient portfolio according to some asset pricing models (such as CAPM), a lot of effort has been put into developing strategies that outperform it (referred colloquially as `beating the market'). Mainstream portfolio theory attempts to achieve this by building asset pricing models that explain and forecast stock returns using economic and technical variables (see \cite{CK06} for an exposition for practitioners). Our approach is more closely aligned to Stochastic Portfolio Theory (SPT) introduced by Fernholz (see \cite{F02, FKsurvey}) who showed that it is possible to build explicit portfolios which beat the market under minimal and realistic assumptions on the behavior of equity markets. Working under a continuous time It\^{o} process model, the theory identifies two fundamental sufficient conditions in terms of the market weight process $\{\mu(t)\}$: {\it diversity} and {\it sufficient volatility}. In the simplest setting, the market is diverse if $\mu(t) \in K$ for all $t$, where $K = \{\mu \in \Delta^{(n)}: \max_{1 \leq i \leq n} \mu_i \leq 1 - \delta\}$ for some $\delta > 0$ (this means that the market is never too concentrated), and the market is sufficiently volatile if the eigenvalues of the diffusion matrix of $\{\mu(t)\}$ are suitably bounded below. (Both notions will be generalized and explained in this work.) Call a \textit{portfolio function} to be a map $\pi:\Delta^{(n)} \rightarrow \overline{\Delta^{(n)}}$. Here, if $\mu \in \Delta^{(n)}$ is the current market weight, one chooses the portfolio $\pi(\mu) \in \overline{\Delta^{(n)}}$. Under the conditions of diversity and sufficient volatility, Fernholz in \cite{F99} showed that certain portfolio functions, that he called \textit{functionally generated}, will outperform the market portfolio after a finite (but large) time, with probability one. These portfolios are called {\it relative arbitrages} with respect to the market. A remarkable fact is that these portfolios, being deterministic functions of the current market weights, are independent of the past or any future forecast. However, not all portfolio functions are functionally generated and it is far from clear from Fernholz's proof whether any other portfolio function, that is not functionally-generated, can also beat the market under the conditions of diversity and sufficient volatility. 

\subsection{Statements of main results}
One of the main contributions of our paper is to show essentially that no other portfolio functions, other than those that are functionally generated, can beat the market in the long run without additional assumptions. In contrast with the traditional continuous time set-up of SPT, in this paper time is taken to be {\it discrete}. We stress on discrete time since it allows complete absence of probabilistic assumptions. While traditionally the market weight is modeled as a continuous time semimartingale, here it is taken to be any deterministic sequence in $\Delta^{(n)}$. All relevant definitions in SPT are modified to this discrete time, pathwise set-up. 

We first state some definitions which will be used througout the paper and in the statements of the main results. As noted above, the market is modeled by a sequence $\{\mu(t)\}_{t = 0}^{\infty}$ of market weights with values in $\Delta^{(n)}$. A {\it portfolio function} is a map $\pi: \Delta^{(n)} \rightarrow \overline{\Delta^{(n)}}$. Every time the market weight is $\mu(t) = p$, the investor chooses the portfolio vector $\pi(\mu(t)) = \pi(p)$. Given a portfolio function $\pi$, the value of the corresponding self-financing portfolio will be measured relative to the value of the market portfolio. Consider the quantity
\[
V(t) = \frac{\text{value at time } t \text{ of \$1 invested in the portfolio } \pi}{\text{value at time } t \text{ of \$1 invested in the market portfolio } \mu}
\]
and call it the {\it relative value}. It can be shown (see for example \cite[Lemma 2.1]{PW13}) that
\begin{equation} \label{eqn:relativevalue}
V(0) = 1, \quad V(t + 1) = V(t) \sum_{i = 1}^n \pi_i(\mu(t)) \frac{\mu_i(t + 1)}{\mu_i(t)}
\end{equation}
and $V(t)$ is strictly positive for all $t$. In order to ``beat the market'', we want to choose $\pi$ such that $V(t)$ is large (at least when $t$ is sufficiently large).

In this context, the concept of relative arbitrage in SPT is extended to the notion of {\it pseudo-arbitrage}. 

\begin{defn}[Pseudo-arbitrage] \label{def:pseudo-arbitrage} 
Let $K$ be a subset of $\Delta^{(n)}$. A portfolio function $\pi$ is called a pseudo-arbitrage on $K$ if the following properties hold.
\begin{itemize}
\item[(i)] There exists a constant $C = C(K, \pi) \geq 0$ such that for all sequences of market weight $\{\mu(t)\}_{t = 0}^{\infty}$ taking values in $K$, we have $\log V(t) \geq -C$ for all $t \geq 0$.
\item[(ii)] There exists some sequence $\{\mu(t)\}_{t = 0}^{\infty} \subset K$ along which $\lim_{t \rightarrow \infty} V(t) = \infty$.
\end{itemize}
\end{defn}

Definition \ref{def:pseudo-arbitrage} formalizes some necessary requirements in order that a given portfolio function is guarenteed to outperform the market under diversity and sufficient volatility. First, under the (generalized) diversity condition $\mu(t) \in K$, the portfolio is never allowed to lose more than a fixed amount (property (i)). That is, the downside risk is uniformly bounded below regardless of the market movements in a fixed region. Second, there is a possibility of unbounded gain (property (ii)).  While these properties appear to be rather weak, they impose strong restrictions on the portfolio map. Given $K$ we give two characterizations of pseudo-arbitrage opportunities. First, we characterize pseudo-arbitrages as portfolios that are functionally generated in Fernholz's sense (with a slightly extended definition).

\begin{thm}\label{thm:charac1} 
A portfolio function $\pi$ is a pseudo-arbitrage on an open convex subset $K \subseteq \Delta^{(n)}$ if and only if there exists a concave function $\Phi:\Delta^{(n)}\rightarrow [0,\infty)$ satisfying the following properties:
\begin{enumerate}[(i)] 
\item the restriction of $\Phi$ on $K$ is not affine, 
\item there exists $\varepsilon>0$ such that $\inf_{p\in K}\Phi(p) \ge \varepsilon$, and 
\item for any $p \in K$, the vector  $\pi(p)/p$ of coordinatewise ratios defines a supergradient of the concave function $\log \Phi$ at $p$ (see Proposition \ref{lem:superdiff} for the precise definition). 
\end{enumerate}
If $\pi$ is continuous, then on $K$ it is necessarily given by the formula  
\eq\label{eq:fgeqn}
\pi_i(p)= p_i\left( 1 + D_{e(i)-p} \log \Phi(p)   \right), \quad i=1,2,\ldots,n.
\en
Here $D_{e(i)-p}$ is the one-sided directional derivative in the direction $e(i)-p$, where $e(i)$ is the vector of all zeroes except one at the $i$th coordinate.
\end{thm}

In the above theorem we say that $\pi$ is {\it generated} by $\Phi$. It can be shown (see Proposition \ref{lem:superdiff} below) that Theorem \ref{thm:charac1}(iii) and \eqref{eq:fgeqn} coincide with Fernholz's definition of functionally generated portfolio under his more restrictive smoothness assumption. 
\medskip

Our second characterization establishes a geometric connection by describing pseudo-arbitrages via solutions to a {\it Monge-Kantorovich optimal transport problem}. To recall the general formulation (see \cite{V03, V08} for details), let ${\mathcal{X}}$ and ${\mathcal{Y}}$ be Polish (complete metric) spaces, and let $c: {\mathcal{X}} \times {\mathcal{Y}} \rightarrow \overline{{\Bbb R}}$ be a measurable function called the {\it cost function}. Let ${\mathcal P}$ and ${\mathcal Q}$ be (Borel) probability measures on ${\mathcal X}$ and ${\mathcal Y}$ respectively. A coupling of $({\mathcal P}, {\mathcal Q})$ is a probability measure ${\mathcal R}$ on ${\mathcal X} \times {\mathcal Y}$ whose marginals are ${\mathcal P}$ and ${\mathcal Q}$ respectively. Let $\Pi({\mathcal P}, {\mathcal Q})$ be the collection of all such couplings. The {\it Monge-Kantorovich optimal transport problem} is the problem
\begin{equation} \label{eqn:Kantorovich}
\inf_{{\mathcal{R}} \in \Pi({\mathcal P}, {\mathcal Q}), (X, Y) \sim {\mathcal{R}}} \E_{{\mathcal R}}\left[ c(X, Y)\right].
\end{equation}
Here the notation means that the random element $(X, Y)$ has joint distribution ${\mathcal R}$. A solution of \eqref{eqn:Kantorovich} is called an {\it optimal coupling}. We say that an optimal coupling $(X, Y)$ of \eqref{eqn:Kantorovich} solves the {\it Monge problem} if $Y = F(X)$ is a deterministic function of $X$. The infinum in \eqref{eqn:Kantorovich} is called the {\it value} of the problem.

Now we specialize to the case ${\mathcal{X}} = \overline{\Delta^{(n)}}$ and ${\mathcal{Y}} = [-\infty, \infty)^n$ together with the cost function
\begin{equation} \label{eqn:ourcost}
c(\mu, h) = \log \left(\sum_{i = 1}^n e^{h_i} \mu_i \right).
\end{equation}
The interpretation is that $\mu$ represents the market weight and $h$ represents the deviation of the portfolio vector from the market weight. Given $\mu \in \Delta^{(n)}$ and $h \in [-\infty, \infty)^n \setminus \{(-\infty, \ldots, -\infty)\}$, we may define a portfolio vector $\pi$ corresponding to $\mu$ via a {\it change of measure}:
\begin{equation} \label{eqn:changemeasure}
\frac{\pi_i}{\mu_i} = \frac{1}{\E_{\mu}[ e^h ]} e^{h_i}, \quad i = 1, \ldots, n,
\end{equation}
where $\E_\mu[ e^h ] := \sum_{i=1}^n e^{h_i}\mu_i$. Consider a probability measure ${\mathcal{P}}$ on $\Delta^{(n)}$ and ${\mathcal{Q}}$ on ${\Bbb R}^n$. Suppose the Monge-Kantorovich problem \eqref{eqn:Kantorovich} with cost \eqref{eqn:ourcost} has a solution $\mcal{R}\in \Pi\left( \mcal{P}, \mcal{Q} \right)$ and value of the problem is finite.

\begin{thm}\label{thm:charac2}
Let $K \subset \Delta^{(n)}$ and suppose $F: K \rightarrow [-\infty, \infty)^n \setminus \{(-\infty, \ldots, -\infty)\}$ is a map such that $(\mu, F(\mu))$ belongs to the support of ${\mcal{R}}$ for all $\mu \in K$, i.e., $(\mu, F(\mu))$ is a selection of the support of ${\mathcal{R}}$. Define a portfolio $\pi$ on $K$ by \eqref{eqn:changemeasure} with $h = F(\mu)$. Then there exists a concave function $\Phi:\Delta^{(n)}\rightarrow [0,\infty)$ such that part (iii) of Theorem \ref{thm:charac1} holds. Thus, $\pi$ is a pseudo-arbitrage on $K$ whenever $K$ is an open convex set and conditions (i) and (ii) in Theorem \ref{thm:charac1} hold.

Conversely, suppose $\pi$ is a pseudo-arbitrage over a subset $K$ of $\Delta^{(n)}$. Suppose that $\left\{ \log (\pi(p) / p ), p \in K\right\}$ is coordinatewise bounded below. Define $h = T(\mu)$ as a function of $\mu$ via $h_i = \log \left( \pi_i(\mu) / \mu_i\right)$ and consider the coupling $(\mu, h)$. For any probability measure ${\mathcal P}$ on $K$, let ${\mathcal Q}$ be the distribution of $h(\mu)$ when $\mu\sim \mcal{P}$. Then the coupling $(\mu, h)$ solves the transport problem \eqref{eqn:Kantorovich} with cost \eqref{eqn:ourcost}.
\end{thm}

For portfolio functions with  strictly positive weights (this corresponds to the above formulation where $h$ takes values in ${\Bbb R}^n$), there is an alternative formulation in terms of the {\it exponential coordinate system} of the unit simplex $\Delta^{(n)}$; see \cite[Example 2.4]{A07}. We view $\Delta^{(n)}$ as an $(n-1)$-dimensional exponential family of probability distributions on $n$ atoms. For $\mu \in \Delta^{(n)}$, we let $\theta = (\theta_1, ..., \theta_{n-1})$ be its exponential coordinates given by
\begin{equation} \label{eqn:definetheta}
\begin{split}
\theta = \iota(\mu) := \left(\log \frac{\mu_1}{\mu_n}, \ldots,  \log \frac{\mu_{n-1}}{\mu_n}\right).
\end{split}
\end{equation}
The map $\iota: \Delta^{(n)} \mapsto {\Bbb R}^{n-1}$ defined by \eqref{eqn:definetheta} is a global coordinate system of the manifold $\Delta^{(n)}$. We can now represent an arbitrary point of $\Delta^{(n)}$ by
\begin{equation} \label{eqn:thetatop}
\mu = \iota^{-1}(\theta) = \left( e^{\theta_1 -\psi(\theta)}, ..., e^{\theta_{n-1} -\psi(\theta)}, e^{-\psi(\theta)}\right), \quad \theta \in {\Bbb R}^{n-1},
\end{equation}
where
\eq\label{eq:psidef}
\psi(\theta) := \log\left(1 + \sum_{i = 1}^{n-1} e^{\theta_i}\right).
\en

Now consider two Borel probability measures $\exP$ and $\exQ$ on $\rr^{n-1}$. Assume, for simplicity of exposition, that $\exP$ is supported on the entire $\rr^{n-1}$. Consider the transport problem $\inf \mathrm{E}\left[ \psi(\theta - \phi)\right]$, where the infimum is taken over all couplings $(\theta, \phi)$ of $(\exP, \exQ)$. Suppose the problem has a solution and let $\left( \theta, \phi(\theta)  \right)$ be some selection from the support of the optimal coupling. Define a portfolio function $\pi:\Delta^{(n)}\rightarrow \Delta^{(n)}$ by the recipe
\eq\label{eqn:portfexpcoord}
\pi\left( \iota^{-1}(\theta) \right)= \iota^{-1} \left( \theta - \phi(\theta) \right).
\en
We will show that this is an equivalent formulation of the transport construction of a portfolio as given in \eqref{eqn:changemeasure}. The advantage of this formulation is that now the transport is on an Euclidean space with a strictly convex cost function. The function $\phi(\theta)$ represents the negative shift, in exponential coordinates, to go from $\mu$ to $\pi$.

The above results allow us to view functionally generated portfolios as maps that minimize the ``total cost'' of assigning portfolio weights to the market weights. This provides an elegant geometric approach and suggests a natural optimization problem for functionally generated portfolios. Namely, given prior beliefs regarding the possible market weights in the future (represented by ${\mathcal{P}}$ or $\widetilde{{\mathcal{P}}}$) and a collection of portfolio weights to be chosen from (represented by ${\mathcal{Q}}$ or $\widetilde{{\mathcal{Q}}}$), the investor can solve an optimal transport problem to obtain a functionally generated portfolio. The transport problem itself is apparently new and its solutions are characterized by Theorem \ref{thm:charac1} and Theorem \ref{thm:charac2}. The transport problem can be solved explicitly for two stocks $(n = 2)$ and in Section \ref{sec:example} we provide several examples using real data.

Both characterizations follow from a novel property of portfolios we call {\it multiplicative cyclical monotonicity} and will be developed in Section \ref{sec:MCM}. Intuitively, this property requires that the portfolio does not underperform the market if the market weight goes over any discrete cycle in the unit simplex.


Apart from the references mentioned above, the following articles are closely related to our current work. The roles of diversity and sufficient volatility in relative arbitrage have been studied by authors such as Fernholz and Karatzas \cite{FK05} and Fernholz, Karatzas, and Kardaras \cite{FKK05}. Functionally generated portfolio is used by Banner and Fernholz \cite{BF08} to prove that sufficient volatility of the smallest stock implies the existence of short term relative arbitrages. Generalizations of functionally generated portfolios to stochastic generating functions and application to statistical arbitrage is studied in Strong \cite{S12}. 
The discrete time set-up and the information-geometric flair is a continuation of Pal and Wong \cite{PW13}. 
The case where the benchmark is a functionally generated portfolio (such as the equal-weighted portfolio) instead of the market is studied in \cite{W14}; a shape-constrained optimization problem for functionally generated portfolios is also formulated. A recent trend in mathematical finance is to study robust pricing and hedging of contingent claims under model uncertainty or model-free assumptions. For example, Fernholz and Karatzas \cite{FK11} characterizes optimal relative arbitrage when the covariance matrix is uncertain. The theory of optimal transport also arises in this context, see, for example, Beiglb\"ock et al{.} \cite{BHP13} and Beiglb\"ock and Juillet \cite{BJ14}, although the motivation is quite different from this paper. 

\subsection{Notations}
Let $n \geq 2$ be the number of stocks of the market. The vector $\mu(t) = (\mu_1(t), ..., \mu_n(t))$ of market weights, as defined by \eqref{eqn:marketweight}, takes value in the open unit simplex $\Delta^{(n)}$ in ${\Bbb R}^n$. All topological and measure-theoretic aspects will be relative to the unit simplex (with the Euclidean topology). Here and everywhere following, for any vectors $a$ and $b$ in ${\Bbb R}^n$, we let $\langle a, b \rangle$ be their Euclidean inner product and $\|a - b\|$ be the Euclidean distance. If $b$ has nonzero components, we denote by $a / b$ the vector of the componentwise ratios $a_i / b_i$. Also we let $a \cdot b$ be the coordinate-wise product. A tangent vector of $\Delta^{(n)}$ is a vector $v \in {\Bbb R}^n$ satisfying $\sum_{i = 1}^n v_i = 0$. We denote by $T\Delta^{(n)}$ the vector space of tangent vectors of $\Delta^{(n)}$. 


\subsection*{Acknowledgment} We are grateful to Prof.~Walter Schachermayer for a thorough reading of the manuscript and suggesting numerous comments for improvement. We also thank the anonymous reviewers for detailed comments about the presentation.


\section{Multiplicative cyclical monotonicity} \label{sec:MCM}
In this section we characterize pseudo-arbitrage portfolios using a property of multivariate functions we call {\it multiplicative cyclical monotonicity}. The main arguments are convex analytic and we begin by reviewing some basic notions. A standard reference of convex analysis is the book \cite{R70} by Rockafellar.

\subsection{Convex analysis on the unit simplex}
A function $\Phi: \Delta^{(n)} \rightarrow {\Bbb R}$ is concave if for any $p, q \in \Delta^{(n)}$ and any $0 \leq \alpha \leq 1$, we have
\[
\Phi(\alpha p + (1 - \alpha) q) \geq \alpha \Phi(p) +  (1 - \alpha) \Phi(q).
\]
Let $\Phi$ be a concave function on $\Delta^{(n)}$. The superdifferential of $\Phi$ is a multi-valued function $\partial \Phi$ from $\Delta^{(n)}$ to $T\Delta^{(n)}$, the set of tangent vectors of $\Delta^{(n)}$. For $p \in \Delta^{(n)}$, a tangent vector $\xi$ belongs to $\partial \Phi(p)$ if
\begin{equation} \label{eqn:superdiff}
\Phi(p) + \langle \xi, q - p \rangle \geq \Phi(q)
\end{equation}
for all $q \in \Delta^{(n)}$. It is well known that $\partial \Phi(p)$ is a non-empty compact convex set. The elements of $\partial \Phi(p)$ are called supergradients (of $\Phi$ at $p$). Intuitively, each supergradient defines via the left hand side of \eqref{eqn:superdiff} a supporting hyperplane of $\Phi$ at $(p, \Phi(p))$. 

We will also consider differentiable functions on $\Delta^{(n)}$ and their derivatives. Implicitly, we regard $\Delta^{(n)}$ as an $(n - 1)$-dimensional manifold embedded in ${\Bbb R}^n$. A global coordinate system is the projection map $\varphi(p) = (p_1, \ldots, p_{n - 1})$. For example, a function $\Phi$ on $\Delta^{(n)}$ is of class $C^k$ ($k$ times continuously differentiable) if the push forward $\Phi \circ \varphi^{-1}$ is of class $C^k$ on the open set $\varphi(\Delta^{(n)})$ in ${\Bbb R}^{n - 1}$. If a concave function is differentiable at $p$, the superdifferential $\partial \Phi(p)$ is a singleton containing the usual gradient.

For $i = 1, \ldots , n$, let $e(i) = (0,  \ldots , 1, \ldots , 0)$ be the vertex of $\Delta^{(n)}$ in the $i$-th direction. If $\Phi$ is a real or vector-valued function on $\Delta^{(n)}$, $p \in \Delta^{(n)}$, and $v \in T\Delta^{(n)}$, we denote by $D_v\Phi(p)$ the directional derivative of $\Phi$ at $p$ in the direction $v$ whenever the limit exists. Explicitly, it is defined by
\[
D_v \Phi(p) = \lim_{h \downarrow 0} \frac{\Phi(p + hv) - \Phi(p)}{h}.
\]
It is well known (\cite[Theorem 23.1]{R70}) that if $\Phi: \Delta^n \rightarrow {\Bbb R}$ is concave, then $D_v\Phi$ exists as a finite limit. If $F$ is vector-valued and differentiable, then $D_{\cdot} F(p)$ is the differential map (push forward of tangent vectors).

\subsection{Multiplicative cyclical monotonicity}
By a {\it cycle} in the unit simplex we mean a finite sequence $\{\mu(t)\}_{t = 0}^{m+1} \subset \Delta^{(n)}$ with $\mu(m+1)=\mu(0)$.

\begin{defn}[Multiplicative cyclical monotonicity (MCM)] \label{def:MCM}
Let $\pi$ be a portfolio or more generally a multi-valued map from $\Delta^{(n)}$ to $\overline{\Delta^{(n)}}$. We say that $\pi$ satisfies multiplicative cyclical monotonicity (MCM) if over any cycle $\{\mu(t)\}_{t = 0}^{m + 1} \subset \Delta^{(n)}$ we have $V(m + 1) \geq 1$, i.e.,
\begin{equation} \label{eqn:MCM}
\prod_{t = 0}^m \left(1 + \left\langle\frac{\pi(\mu(t))}{\mu(t)}, \mu(t + 1) - \mu(t)\right\rangle \right) \geq 1.
\end{equation}
For $\delta > 0$, we say that $\pi$ satisfies $\delta$-MCM if the inequality \eqref{eqn:MCM} holds for all cyclces where the successive jump sizes $\|\mu(t + 1) - \mu(t)\|$ are all less than $\delta$.
\end{defn}

 From Definition \ref{def:pseudo-arbitrage} the minimum relative value of a pseudo-arbitrage is uniformly bounded below whenever $\mu(t)$ takes value in the subset $K$ of $\Delta^{(n)}$. Since we are allowed to choose any path in $K$, this is the case only if the portfolio does not underperform the market when the market weights goes through a cycle and return to an earlier position. Here is an immediate consequence of Definition \ref{def:MCM}.

\begin{lemma} \label{lem:MCMfails}
Suppose $\pi$ is a portfolio which fails the MCM property. Then there exists a market weight sequence $\{\mu(t)\}_{t = 0}^{\infty}$ such that $\mu(t)$ takes values in a finite subset of $\Delta^{(n)}$ and $V(t) \rightarrow 0$ as $t \rightarrow \infty$.
\end{lemma}
\begin{proof}
Suppose that the MCM property breaks down for some cycle $\left\{\widetilde{\mu}(t)\right\}_{t = 0}^{m + 1}$ in $\Delta^{(n)}$. Then
\eq\label{eq:mcmfail}
\eta:=\prod_{t=0}^m \left( 1 + \iprod{\frac{\pi\left(\widetilde{\mu}(t+1)\right)}{\widetilde{\mu}(t)}, \widetilde{\mu}(t+1)-\widetilde{\mu}(t)} \right)  < 1,
\en
where $\widetilde{\mu}(m+1) = \widetilde{\mu}(0)$. Consider the market weight sequence $\{\mu(t)\}_{t = 0}^{\infty}$ defined by
\[
\mu(t) = \widetilde{\mu}(k), \quad \text{if}\; t = k \pmod{m+1}.
\]
In other words, the market cycles through the sequence $\left\{\widetilde{\mu}(t)\right\}_{t = 0}^m$. Clearly $\mu(t)$ takes values in a finite set. By \eqref{eq:mcmfail}, we have $V(k(m+1)) = \eta^k$ which tends to zero as $k$ tends to infinity. It is easy to see that $V(t) \rightarrow 0$ as well.
\end{proof}

The sequence $\{\mu(t)\}_{t = 0}^{\infty}$ in the proof of Lemma \ref{lem:MCMfails} is by all standards diverse and sufficiently volatile, yet the relative value of the portfolio $\pi$ goes to zero. It follows that the MCM property is necessary in order that a portfolio outperforms the market in {\it all} diverse and sufficiently volatile markets. In practice, it is of course unlikely that market weights travels along a cycle. The MCM property is a desirable theoretical property that a portfolio may or may not satisfy. It is clear that MCM implies $\delta$-MCM for any $\delta > 0$. The definition above is a multiplicative form of the classical cyclical monotonicity property in convex analysis \cite[Section 24]{R70}.

Our first result gives the connection between the MCM property and concave functions.

\begin{prop}\label{thm:MCM}
Let $\pi$ be a portfolio or more generally a multi-valued map from $\Delta^{(n)}$ to $\overline{\Delta^{(n)}}$. 
\begin{enumerate}
\item[(i)] $\pi$ satisfies MCM if and only if there exists a concave function $\Phi: \Delta^{(n)} \rightarrow (0, \infty)$ such that
\begin{equation} \label{eqn:FGineq}
1 + \left\langle \frac{\pi(p)}{p}, q - p \right\rangle \geq \frac{\Phi(q)}{\Phi(p)}, \quad \text{for all $p, q \in \Delta^{(n)}$. }
\end{equation}

\item[(ii)] For any $\delta > 0$, if $\pi$ satisfies $\delta$-MCM, then $\pi$ satisfies MCM.
\end{enumerate}
\end{prop}

\begin{proof}
(i) The proof is an adaptation of the proof of \cite[Theorem 24.8]{R70}. For notational simplicity we assume that $\pi$ is single-valued. Suppose there exists a concave function $\Phi: \Delta^{(n)} \rightarrow (0, \infty)$ such that \eqref{eqn:FGineq} holds. Then over any discrete cycle $\{\mu(t)\}_{t = 0}^{m+1}$ with $\mu(m + 1) = \mu(0)$, we have
\[
\prod_{t = 0}^m \left(1 + \left\langle \frac{\pi(\mu(t))}{\mu(t)}, \mu(t + 1) - \mu(t) \right\rangle\right) \geq \prod_{t = 0}^m \frac{\Phi(\mu(t + 1)}{\Phi\left(\mu(t)\right)} = 1.
\]
The final equality holds since $\mu(m + 1) = \mu(0)$. This shows that $\pi$ satisfies MCM.

Conversely, suppose that $\pi$ satisfies MCM. Fix a point $\mu(0) \in \Delta^{(n)}$ and define a function $\Phi$ on $\Delta^{(n)}$ by
\begin{equation} \label{eqn:defpotential}
\Phi(p) = \inf_{m \geq 0, \{\mu(t)\}_{t = 0}^{m + 1}, \mu(m+1) = p} \left[ \prod_{t = 0}^m \left( 1 + \left\langle \frac{\pi(\mu(t))}{\mu(t)}, \mu(t + 1) - \mu(t) \right\rangle\right) \right].
\end{equation}
Here the infimum is taken over all $m \geq 0$ and all choices of points $\{\mu(t)\}_{t = 0}^{m + 1} \subset \Delta^{(n)}$ such that $\mu(m + 1) = p$. We claim that \eqref{eqn:FGineq} holds. 

Clearly $\Phi$, being the pointwise infimum of a family of affine functions in $p$, is a concave function on $\Delta^{(n)}$. It is clear that $\Phi$ is non-negative, and the MCM property implies that $\Phi(\mu(0)) = 1$. It follows by concavity that $\Phi$ must be everywhere positive on $\Delta^{(n)}$.

To show \eqref{eqn:FGineq}, let $p, q \in \Delta^{(n)}$ be given. Let $\alpha > \Phi(p)$. By definition of $\Phi$, there exists some $m \geq 0$ and a sequence $\{\mu(t)\}_{t = 0}^{m+1}$ with $\mu(m + 1) = p$ such that
\[
\prod_{t = 0}^m \left(1 + \left\langle \frac{\pi(\mu(t))}{\mu(t)}, \mu(t + 1) - \mu(t) \right\rangle \right) < \alpha.
\]
Setting $\mu(m + 2) = q$, we have
\[
\Phi(q) \leq \left(1 + \left\langle \frac{\pi(p)}{p}, q - p \right\rangle\right) \alpha.
\]
The proof of \eqref{eqn:FGineq} is completed by letting $\alpha \downarrow \Phi(p)$.

\medskip

\noindent
(ii) The idea is to repeat the proof of (i) with the additional restriction that the jumps have sizes less than $\delta$. Consider the function $\Phi$ defined by \eqref{eqn:defpotential}, where the infimum is taken over all $m \geq 0$ and all choices of $\{\mu(t)\}_{t = 0}^{m+1}$ where $\|\mu(t + 1) - \mu(t)\| < \delta$. As before, $\Phi$ is a positive concave function on $\Delta^{(n)}$, and \eqref{eqn:FGineq} holds whenever $\|p - q\| < \delta$. Thus 
\begin{equation} \label{eqn:superdiff2}
\Phi(p) + \left\langle \frac{\pi(p)}{p} \Phi(p), q - p \right\rangle \geq \Phi(q), \quad \|p - q\| < \delta.
\end{equation}
This shows that the component of $\frac{\pi(p)}{p} \Phi(p)$ parallel to $\Delta^{(n)}$ (which is a tangent vector) is a superdifferential of the restricted concave function $\Phi|_V$ at $p$, where $V$ is a convex neighborhood of $p \in \Delta^{(n)}$. However, by \cite[Theorem 23.2]{R70}, we have
\[
\partial \Phi(p) = \{\xi \in T\Delta^{(n)}: D_v \Phi(p) \leq \langle \xi, v \rangle, \quad \text{for all } v \in T\Delta^{(n)}\}.
\]
Since the one-sided derivatives of $\Phi$ depends only on the values of $\Phi$ in a neighborhood of $p$, we observe that $\partial \Phi(p) = \partial (\Phi|_V)(p)$ for any convex neighborhood $V$ of $p$. It follows that \eqref{eqn:superdiff2} holds for all $p, q \in \Delta^{(n)}$. Thus by (i) $\pi$ satisfies MCM.
\end{proof}

\begin{rmk}\label{rmk:cxmcm}
If $K$ is any subset of $\Delta^{(n)}$, we can define the MCM property on $K$ by requiring that \eqref{eqn:MCM} holds for all cycles $\{\mu(t)\}_{t = 0}^m$ in $K$. A straightforward modification of the proof of Theorem \ref{thm:MCM} shows that there is a positive concave function $\Phi$ defined on the entire simplex $\Delta^{(n)}$ such that \eqref{eqn:FGineq} holds for all $p, q \in K$. This observation will be useful in the proof of Theorem \ref{thm:mcmarbitrage}.
\end{rmk}

\subsection{Functionally generated portfolios}
In the set-up of Proposition \ref{thm:MCM} we say that $\pi$ is {\it generated} by $\Phi$ and $\Phi$ is a {\it generating function} of $\pi$. Note that if $\Phi$ is a positive concave function on $\Delta^{(n)}$, then $\log \Phi$ is also concave and
\[
\partial \log \Phi(p) = \frac{1}{\Phi(p)} \partial \Phi(p) = \left\{\frac{1}{\Phi(p)} \xi: \xi \in \partial \Phi(p) \right\}.
\]

Our next propositions show that MCM portfolios are functionally generated in the sense of Fernholz (see \cite[Theorem 3.1.5]{F02} for his formulation). 

\begin{prop} \label{lem:superdiff} Let $\Phi$ be a positive concave function on $\Delta^{(n)}$.
\begin{enumerate}
\item[(i)] Suppose the portfolio $\pi$ is generated by $\Phi$. For $p \in \Delta^{(n)}$, the tangent vector $v = (v_1, \ldots, v_n)$ defined by
\begin{equation} \label{eqn:definev}
v_i = \frac{\pi_i(p)}{p_i} - \frac{1}{n} \sum_{j = 1}^n \frac{\pi_j(p)}{p_j}, \quad i = 1, \ldots, n,
\end{equation}
belongs to $\partial \log \Phi(p)$.
\item[(ii)] Conversely, if $v \in \partial \log\Phi(p)$, the vector $\pi = (\pi_1, \ldots, \pi_n)$ defined by
\begin{equation} \label{eqn:definepi}
\frac{\pi_i}{p_i} = v_i + 1 - \sum_{j = 1}^n p_jv_j, \quad i = 1, \ldots, n,
\end{equation}
is an element of $\overline{\Delta^{(n)}}$. 
\end{enumerate}
Moreover, the operations $\pi \mapsto v$ and $v \mapsto \pi$ defined by \eqref{eqn:definev} and \eqref{eqn:definepi} are inverses of each other.
\end{prop}

\begin{rmk} \label{rmk:selection}
In particular, any selection of $\partial \log \Phi$ (a map $v: \Delta^{(n)} \rightarrow T\Delta^{(n)}$ such that $v(p) \in \partial \log \Phi(p)$ for all $p \in \Delta^{(n)}$) defines via \eqref{eqn:definepi} a portfolio generated by $\Phi$. For certain applications it is required that the selection is measurable. By \cite[Theorem 14.56]{RW98}, such a selection always exists.
\end{rmk}

\begin{proof}
(i) Let $p \in \Delta^{(n)}$. By \eqref{eqn:FGineq}, we have
\[
1 + \left\langle \frac{\pi(p)}{p}, q - p \right\rangle \geq \frac{\Phi(q)}{\Phi(p)}
\]
for all $q \in \Delta^{(n)}$. Note that $\frac{\pi(p)}{p}$ is not a tangent vector of $\Delta^{(n)}$. The normalization \eqref{eqn:definev} `projects' $\frac{\pi(p)}{p}$ to $v$ which is a tangent vector. Since $\frac{\pi(p)}{p} - v$ is perpendicular to $T\Delta^{(n)}$, the inner product does not change if $\frac{\pi(p)}{p}$ is replaced by $v$. Hence by \eqref{eqn:superdiff} we have $v \in \partial \log \Phi(p)$.

\medskip
\noindent
(ii) It is easy to verify that $\sum_{i = 1}^n \pi_i = 1$. To see that $\pi_i \geq 0$ for each $i$, let $q - p = t(e(i) - p)$ for $0 < t < 1$ in \eqref{eqn:superdiff}, and we have
\begin{equation*}
\begin{split}
-\Phi(p) &\leq \Phi(p + t(e(i) - p)) - \Phi(p) \quad \text{(since } \Phi(q) > 0\text{)}\\
            &\leq \langle \Phi(p)v, t(e(i) - p) \rangle 
            = t\Phi(p) \left(v_i - \sum_{j = 1}^n p_j v_j\right).
\end{split}
\end{equation*}
Letting $t \uparrow 1$ and dividing both sides by $\Phi(p)$, we get the desired inequality $\pi_i \geq 0$.

\medskip 

That $\pi \mapsto v$ and $v \mapsto \pi$ are inverses of each other can be verified by a direct computation.
\end{proof}

\begin{prop} \label{lem:fgpproperties}
Let $\pi$ be a portfolio on $\Delta^{(n)}$ generated by a concave function $\Phi$.
\begin{enumerate}[(i)]
\item The generating function $\Phi$ is unique up to a positive multiplicative constant.
\item For any $\mu \in \Delta^{(n)}$ and $i = 1, \ldots, n$, we have
\[
1 + D_{e(i) - \mu} \log \Phi(\mu) \leq \frac{\pi_i}{\mu_i} \leq 1 - D_{\mu - e(i)} \log \Phi(\mu).
\]
In particular, if $\Phi$ is differentiable, the portfolio is given by the formula
\begin{equation} \label{eq:fnlygen}
\pi_i = \mu_i \left( 1 + D_{e(i) - \mu} \log \Phi(\mu) \right), \quad i = 1, \ldots, n.
\end{equation}

\item If $\pi$ is continuous, then $\Phi$ is continuously differentiable. More generally, if $\pi$ is of class $C^k$, then $\Phi$ is of class $C^{k+1}$.

\item If $\Phi: \Delta^{(n)} \rightarrow (0, \infty)$ is concave and differentiable, and we define $\pi$ by \eqref{eq:fnlygen}, then $\pi$ is generated by $\Phi$. In particular, $\pi(\mu) \in \overline{\Delta^{(n)}}$ for all $\mu$.
\end{enumerate}
\end{prop}

\begin{proof}
(i) Suppose $\pi$ is generated by $\Phi_1$ and $\Phi_2$ which generate $\pi$. Let $p, q \in \Delta^{(n)}$ and consider the line segment $\ell$ from $p$ to $q$. Consider the restrictions of $\log \Phi_i$ to $\ell$, denoted by $\left. \log \Phi_i \right|_{\ell}$. They can be parametrized as one-dimensional concave functions. In particular, they are differentiable on $\ell$ except at most for countably many points on $\ell$. By Proposition \ref{lem:superdiff}, the vector $\pi(\mu)/\mu$ defines a supporting hyperplane of the log generating function. It follows that $\log \Phi_1$ and $\log \Phi_2$ have parallel supporting hyperplanes at all points of $\Delta^{(n)}$. In particular, the derivatives of $\left. \log \Phi_1 \right|_{\ell}$ and $\left. \log \Phi_2 \right|_{\ell}$ agree almost everywhere on $\ell$. By the fundamental theorem of calculus for concave functions \cite[Corollary 24.2.1]{R70}, we have
\[
\log \Phi_1(q) - \log \Phi_1(p) = \log \Phi_2(q) - \log \Phi_2(p).
\]
Since $p$ and $q$ are arbitrary, $\Phi_2 / \Phi_1$ is a positive constant.
\medskip

\noindent (ii) By definition of $\pi$, for $h \in {\Bbb R} \setminus \{0\}$ small enough such that $\mu + h(e(i) - \mu) \in \Delta^{(n)}$, the superdifferential inequality \eqref{eqn:FGineq} gives
\begin{equation} \label{eq:subdiffineq}
1 + \left\langle \frac{\pi(\mu)}{\mu}, h(e(i) - \mu) \right\rangle \geq \frac{\Phi(\mu + h(e(i) - \mu))}{\Phi(\mu)}.
\end{equation}
Note that the inner product is given by
\[
\left\langle \frac{\pi(\mu)}{\mu}, h(e(i) - \mu) \right\rangle = h \left( \frac{\pi_i}{\mu_i} - 1 \right).
\]
Taking log on both sides of \eqref{eq:subdiffineq}, we have
\[
\log\left(1 + h \left( \frac{\pi_i}{\mu_i} - 1 \right)\right) \geq \log \Phi(\mu + h(e(i) - \mu)) - \log \Phi(\mu).
\]
Dividing by $h$ and taking the limits as $h \downarrow 0$ and $h \uparrow 0$, we obtain the desired inequalities. Formula \eqref{eq:fnlygen} is proved by noting that if $\Phi$ (and hence $\log \Phi$) is differentiable, for every tangent vector $v$ we have $D_v \log \Phi = - D_{-v} \log \Phi$.

\medskip
\noindent
(iii) Suppose that $\pi$ is continuous. It follows that $p \mapsto \frac{\pi(p)}{p}$ is a continuous selection of the superdifferential of $\log \Phi$. By \cite[Proposition 4]{R88} $\log \Phi$, and hence $\Phi$, is differentiable on $\Delta^{(n)}$. By \cite[Corollary 25.5.1]{R70}, $\Phi$ is actually continuously differentiable.

If $\pi$ is of class $C^k$ where $k \geq 1$, we already know $\Phi$ is differentiable. In terms of the coordinate system $\varphi(\mu) = (\mu_1, \ldots, \mu_{n-1})$, for $i = 1, \ldots, n-1$ we have
\begin{equation*}
\begin{split}
& \frac{\partial}{\partial {\mu}_i} \left( \log \Phi \circ \varphi^{-1} \right)(\mu_1, \ldots, \mu_{n-1}) = D_{e(i) - e(n)} \log \Phi(\mu) \\
   &= D_{e(i) - \mu} \log \Phi(\mu) - D_{e(n) - \mu} \log \Phi(\mu) = \frac{\pi_i}{\mu_i} - \frac{\pi_n}{\mu_n}
\end{split}
\end{equation*}
which is of class $C^k$. Hence $\log \Phi$, and hence $\Phi$, is of class $C^{k+1}$.

\medskip
\noindent
(iv) Suppose $\Phi$ is differentiable and define $\pi$ by \eqref{eq:fnlygen}. To show $\pi_i \geq 0$, fix $\mu \in \Delta^{(n)}$ and $1 \leq i \leq n$, and consider the restriction of $\Phi$ to the segment $[\mu, e(i))$. The inequality $\pi_i \geq 0$ is equivalent to $D_{e(i) - \mu} \Phi(\mu) \geq - \Phi(\mu)$ and follows from concavity and positivity of $\Phi$. To verify that $\pi \in \overline{\simp^{(n)}}$ we only need to show that $\sum_{i = 1}^n \pi_i = 1$, or $\sum_{i=1}^n \mu_i D_{e(i)-\mu} \log \Phi(\mu) =0$. This follows from the identity $\sum_{i = 1}^n \mu_i(e(i) - \mu) = 0$ and linearity of the directional derivative for differentiable functions. Finally, $\pi$ is generated by $\Phi$ by (ii).
\end{proof}

\subsection{L-divergence}
The concavity of the generating function has a financial meaning which will become clear in Lemma \ref{lem:FernholzDecomp} below.

\begin{defn}[L-divergence] \label{def:discreteenergy}
Let $\pi$ be a portfolio generated by a concave function $\Phi$. The {\it L-divergence} of the pair $(\Phi, \pi)$ is defined by
\begin{equation} \label{eq:discreteenergy}
\diverge\left( q \mid p \right) := \log\left(  1 + \iprod{\frac{\pi(p)}{p}, q - p} \right) - \left( \log \Phi(q) - \log \Phi(p) \right), \quad p, q \in \Delta^{(n)}.
\end{equation}
\end{defn}

We use this terminology because the expression \eqref{eq:discreteenergy} can be regarded as a logarithmic version of the widely used {\it Bregman divergence} (see \cite{AC10}). It is clear from \eqref{eqn:FGineq} that $\diverge(q \mid p) \ge 0$ and is strictly positive unless $\Phi$ is affine on the straight line joining $p$ and $q$. When $\Phi$ is strictly concave, $T(q \mid p) = 0$ only if $q = p$. In general, $\diverge$ is asymmetric and does not define a metric.

As an example, fix $\pi \in \overline{\Delta^{(n)}}$ and consider the geometric mean $\Phi(p)= \prod_{i=1}^n p_i^{\pi_i}$. This is a concave function which generates the constant-weighted portfolio $\pi$ (\cite[Example 3.1.6]{F02}). We have
\eq\label{eq:freeenergy}
\diverge(q\mid p) = \log\left( \sum_{i = 1}^n \pi_i \frac{q_i}{p_i}   \right) -\sum_{i=1}^n \pi_i \log\left( \frac{q_i}{p_i}  \right).
\en
We call this the {\it discrete excess growth rate} $\gamma_{\pi}^*$, see \cite[Definition 2.2]{PW13}.

\begin{rmk}
The L-divergence is defined for a {\it pair} $(\Phi, \pi)$ because in general the portfolio generated by a given concave function is not unique. Whenever $\Phi$ is not differentiable at $\mu$ any choice of supergradient defines a portfolio vector. Since a concave function is differentiable almost everywhere, the portfolios agree almost everywhere on $\Delta^{(n)}$.
\end{rmk}

The L-divergence features in the following decomposition formula which is a discrete time analogue of \cite[Theorem 3.1.5]{F02}.

\begin{lemma} \label{lem:FernholzDecomp}
Let $\pi$ be generated by a positive concave function $\Phi$. Let $T$ be the L-divergence functional of the pair $(\Phi, \pi)$. Then the relative value process $V(t)$ satisfies the decomposition
\eq\label{eq:comparevt}
\log V(t)= \log \frac{\Phi\left( \mu(t) \right)}{\Phi\left( \mu(0) \right)} + A(t),
\en
where $A(t) = \sum_{k=0}^{t-1} \diverge\left(\mu(k+1) \mid \mu(k)  \right)$ is non-decreasing. Moreover, $\Phi$ is affine if and only if $A(t) \equiv 0$ for all market weight sequences.
\end{lemma}

Following Fernholz, we call $A(t)$ the {\it drift process}.

\begin{proof}
The decomposition formula \eqref{eq:comparevt} follows directly from the definitions. By \eqref{eqn:relativevalue} and \eqref{eq:discreteenergy}, we have
\begin{equation*}
\begin{split}
T(\mu(t + 1)\mid\mu(t)) &= \log \left(1 + \left\langle \frac{\pi(\mu(t))}{\mu(t)}, \mu(t + 1) - \mu(t) \right\rangle \right) - \log \frac{\Phi(\mu(t+1))}{\Phi(\mu(t))} \\
                    &= \log \frac{V(t+1)}{V(t)} - \log \frac{\Phi(\mu(t+1))}{\Phi(\mu(t))}.
\end{split}
\end{equation*}
Thus  \eqref{eq:comparevt} follows by summing over $t$ and rearranging. Since $T\left(q \mid p \right) \geq 0$, $A$ is non-decreasing.

It is clear that $A(t) \equiv 0$ if $\Phi$ is affine. Conversely, suppose $A(t) \equiv 0$ for all market weight sequences. Then $\Phi$ is affine on the line segment $[p, q]$ as $T\left(p \mid p \right) = 0$. Since $p$ and $q$ are arbitrary, $\Phi$ is affine on $\Delta^{(n)}$.
\end{proof}

Whenever $\Phi$ is not affine on the line segment $[\mu(t), \mu(t + 1)]$, the L-divergence $T\left(\mu(t + 1) \mid \mu(t)\right)$ is strictly positive. The drift process $A(t)$ measures the cumulative amount of market volatility captured by the portfolio in the time interval $[0, t]$. We say that the market is {\it sufficient volatility} for $\pi$ if $A(t) \uparrow \infty$ as $t \rightarrow \infty$.

\subsection{Characterization of pseudo-arbitrage}
If $\pi$ is not MCM, Lemma \ref{lem:MCMfails} shows that there is a sequence of market weights along which the relative value goes to zero. This is a `global' property in the sense that the faulty path might require big jumps.  Our next result shows that failing to be MCM is also a `local' property in the sense that the jumps could be as small as we wish. In fact, it could be completely localized around a single point. 

\begin{thm}\label{thm:mcmarbitrage}
Let $\pi$ be a portfolio map from $\Delta^{(n)}$ to $\overline{\Delta^{(n)}}$ which fails the MCM property.
\begin{enumerate}[(i)]
\item For any $\delta > 0$, there is a sequence $\{\mu(t)\}_{t = 0}^{\infty}$ of market weights such that $\|\mu(t + 1) - \mu(t)\| < \delta$ for all $t$ and $V(t)$ goes to zero as $t \rightarrow \infty$. Thus $\pi$ cannot be a pseudo-arbitrage over any set containing the path.
\item For any $\delta >0$, there exists $p \in \Delta^{(n)}$ such that the sequence in (i) can be chosen to lie entirely within the Euclidean ball of radius $\delta$ around $p$.
\end{enumerate}
\end{thm}

\begin{proof}
Part (i) has been proved in Lemma \ref{lem:MCMfails}. That the sequence can be chosen with jump size less than $\delta$ follows from Proposition \ref{thm:MCM}(ii).

\medskip

To prove (ii), we will show that given $\delta > 0$, there is a point $p\in \Delta^{(n)}$ such that the MCM property fails inside the Euclidean ball of radius $\delta$ around $p$. Then we may repeat the proof of (i) in this ball. This will be achieved by a method of contradiction using the following claim.

\smallskip

\nin\tbf{Claim.} Suppose there exists $\delta >0$ such that for any $p\in \Delta^{(n)}$, the MCM property holds over any choice of points selected within a ball of radius $\delta$ around $p$. Then the MCM property holds on $\Delta^{(n)}$.

\smallskip

To prove the above claim let us recall the notions of line integrals and conservative vector fields. Let $\gamma$ be a piecewise linear curve in $\simp^n$ indexed by a closed interval, say $[0,1]$. The curve will be called a loop if $\gamma(0)=\gamma(1)$. The line integral of the vector field $w(\mu) := \pi(\mu)/\mu$ over any $\gamma$ will be denoted by
\[
I_{\gamma}(w) :=\int_{\gamma} \frac{\pi(\mu)}{\mu} d\mu= \int_0^1 \iprod{\frac{\pi(\mu(t))}{\mu(t)},\mu'(t)} dt.
\]
The line integral does not depend on parametrization, except for the orientation. By a slight abuse of notation, the line from any $a$ to any $b$ in $\simp^n$, irrespective of parametrization, will be denoted by $[a,b]$.

Let $p\in \Delta^{(n)}$ and let $B_\delta(p) = \{q \in \Delta^{(n)}: \|q - p\| < \delta\}$. Consider any loop $\gamma$ whose range is contained in $B_\delta(p)$. Then we have
\eq\label{eq:consvec}
I_{\gamma}(w)=0.
\en
In other words, the vector field $w$ is locally conservative restricted to every $B_\delta(p)$. To see \eqref{eq:consvec}, we use the fact that $\pi$ satisfies MCM over $B_\delta(p)$. Therefore, by Theorem \ref{thm:MCM}, there is a positive concave function $\Phi$ on $\simp^n$ which generates $\pi$ on $B_{\delta}(p)$. Consider any line $\ell = [p_1, p_2]$ contained in $B_\delta(p)$. By \cite[Theorem 24.2]{R70}, we have $I_{\ell}(\pi/\mu)=\log \Phi(p_2) - \log \Phi(p_1)$. Thus \eqref{eq:consvec} holds for any piecewise linear loop in $B_{\delta}(p)$.

\medskip

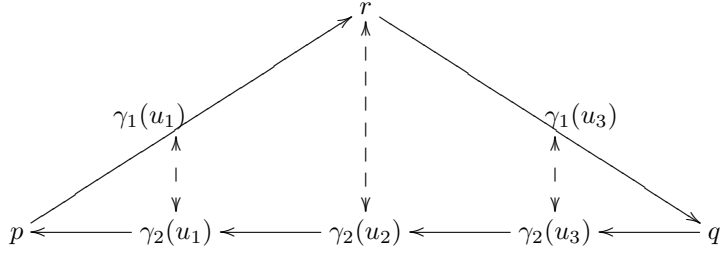
\begin{figure}[t]
\centerline{
\xymatrix@=2.8em{
& & r \ar[drdr] \ar@{<-->}[dd] & &\\
 & \gamma_1(u_1)\qquad \ar@{<-->}[d] & & \qquad \gamma_1(u_3)  \ar@{<-->}[d]& \\
p \ar[urur] & \gamma_2(u_1) \ar[l] &\gamma_2(u_2) \ar[l] & \gamma_2(u_3) \ar[l] & q \ar[l]
}
}
\caption{Decomposing a loop as a union of local loops. Here $r=\gamma_1(u_2)$.}
\label{fig:loopdecomp}
\end{figure}

We now show that any locally bounded and conservative vector field over $\Delta^{(n)}$ must be \textit{globally} conservative. While this statement is well known for smooth vector fields, we only assume that $\pi / \mu$ is measurable and locally bounded, and the resulting potential $\log \Phi$ is not necessarily differentiable. Since we are unable to find a reference for this result, we will give a sketch of proof and refer the reader to \cite[Proof of Theorem 8]{PW14} for more details.  

Let $w(\mu) =\pi(\mu)/\mu$ be locally conservative in the sense of \eqref{eq:consvec}. Fix $p, q \in \Delta^{(n)}$ and consider two piecewise linear curves $\gamma_1$ and $\gamma_2$ from $p$ to $q$. We will be done once we show
\eq\label{eq:linefree}
\int_{\gamma_1} w(\mu)d\mu = \int_{\gamma_2} w(\mu)d\mu.
\en
Without loss of generality, we may assume that $\gamma_2(t)= (1-t)p + t q$.

In fact, we can assume that $\gamma_1$ has exactly three corners $p,r,q$ and is a concatenation of $[p,r]$ and $[r,q]$ (we call such curves triangular). This is because once we establish \eqref{eq:linefree} for such triangular curves, we can inductively eliminate corners in any other $\gamma_1$ and establish \eqref{eq:linefree} in general.

\medskip

For the rest of the argument we assume that $\gamma_1$ is triangular and $\gamma_2$ is $[p,q]$. Assume both $\gamma_1$ and $\gamma_2$ are indexed by $[0,1]$.

\smallskip

We first suppose that $\sup_{0\le t \le 1} \norm{\gamma_1(t) - \gamma_2(t)} < \delta/2$. In this case, choose points $u_0=0 <u_1 < u_2, \ldots$ in $[0,1]$ such that their images on $\gamma_2$ are a sequence of equidistant points with successive distance less than $\delta/2$. Now add lines between $\gamma_1(u_i)$ and $\gamma_2(u_i)$. Now consider each loop which is formed by the $4$ oriented lines $[\gamma_2\left(u_{i+1}\right), \gamma_2\left(u_{i}\right)]$, $[\gamma_2\left( u_{i} \right), \gamma_1\left( u_{i} \right)]$, $[\gamma_1\left( u_i \right), \gamma_1\left( u_{i+1} \right)]$, and $[\gamma_1\left(  u_{i+1} \right), \gamma_2\left( u_{i+1} \right)]$. See Figure \ref{fig:loopdecomp}.

By the triangle inequality for Euclidean distance it follows that the loop lies entirely inside $B_{\delta}\left( \gamma_2\left(u_i\right) \right)$. Hence, by our assumption on local conservation, the integrals of $w$ over these loops are zero. However, the sum of the integrals over all these loops is precisely the integral of $w$ over the concatenation of lines $\gamma_1$ and $-\gamma_2$. Therefore this integral is zero, proving \eqref{eq:linefree}.

\smallskip

It can be shown by means of a simple geometric argument that any other case can be reduced to Case 1 above (see \cite{PW14} for details). Now that we have shown that $w$ is globally conservative, we can unambiguously define a function $\Phi$ on $\Delta^{(n)}$ by fixing some $p_0\in \Delta^{(n)}$ and defining
\begin{equation} \label{eqn:potential}
\log \Phi(p)=\int_\gamma \frac{\pi}{\mu} d\mu, \quad p\in \Delta^{(n)},
\end{equation}
where the integral is over any piecewise linear curve from $p_0$ to $p$. Over any $B_\delta(p)$, the function $\Phi$ must coincide (up to a constant) with the concave function resulting from the local MCM property of the vector field $w$.
Thus, $\Phi$ is locally concave on $\Delta^{(n)}$ and hence it is concave (see \cite[page 58]{H07}) and generates $\pi$. This shows that $\pi$ is MCM over $\Delta^{(n)}$ and this completes the proof of the theorem.
\end{proof}

\begin{proof}[Proof of Theorem \ref{thm:charac1}]
Sufficiency follows from the decomposition formula \eqref{eq:comparevt}. By condition (ii), the first term on the right is bounded uniformly over all sequences $\{\mu(t)\}_{t = 0}^{\infty} \subset K$ (note that a positive concave function on $\Delta^{(n)}$ is bounded above). Since $K$ is open, the L-divergence is identically zero on $K \times K$ if and only if the function $\Phi$ is affine. If $\Phi$ is non-affine, one can clearly choose a sequence such that the accumulated L-divergence goes to infinity, and so $V(t) \rightarrow \infty$ along that sequence. Equation \eqref{eq:fgeqn} is taken from Proposition \ref{lem:fgpproperties}(ii).

To prove necessity, we first note that $\pi$ cannot be a pseudo-arbitrage if there is a market cycle in $K$ over which the MCM property fails. By Remark \ref{rmk:cxmcm}, there exists a positive concave function $\Phi$ on $\Delta^{(n)}$ such that Theorem \ref{thm:charac1}(iii) holds. Since $\pi$ is functionally generated, we may apply Lemma \ref{lem:FernholzDecomp}. If $\Phi$ is affine over $K$ then the process $A$ is zero, and $V(t)$ is bounded above if $\mu(t) \in K$ for all $t$. This prevents $\pi$ from being a pseudo-arbitrage.

Finally, suppose zero is a limit point of the set $\Phi(K)$. As a positive concave function on $\Delta^{(n)}$, $\Phi$ can be extended continuously to $\overline{\Delta^{(n)}}$. We can thus find a point $q$ in the closure such that $\Phi(q)=0$. Fix a point $p\in K$ and let $\{\lambda(t)\}_{t = 0}^{\infty}$ be a strictly increasing sequence in $[0, 1)$ converging to $1$. Let $\{\mu(t)\}_{t \geq 0}$ be the sequence of market weights defined by
\[
\mu(t) = (1 - \lambda(t)) p + \lambda(t) q.
\]
We choose $\lambda(t)$ such that $\log \Phi$ is differentiable at $\mu(t)$ for all $t$. Since $\log \Phi$ is differentiable, we have
$\left\langle \frac{\pi(p)}{p}, q - p \right\rangle = D_{q - p} \log \Phi(p)$. It follows that
\[
\iprod{\frac{\pi(\mu(t))}{\mu(t)}, \mu(t + 1)-\mu(t)} = D_{q - p} \log \Phi(\mu(t)) \left(\lambda(t + 1) - \lambda(t)\right).
\] 
Using the elementary inequality $\log(1+x) \le x$ for $x>-1$, we get
\begin{equation} \label{eqn:uppderbound}
\sum_{t=0}^\infty \log\left( 1 +  \iprod{\frac{\pi(\mu(t))}{\mu(t)}, \mu(t+1)-\mu(t)}   \right) \le \sum_{t=0}^\infty   D_{q - p} \log \Phi(\mu(t))\left(\lambda(t + 1) - \lambda(t)\right).
\end{equation}
Comparing the right hand side of \eqref{eqn:uppderbound} with the line integral
\[
\int_{[p, q]} \frac{\pi}{\mu} d\mu = \log \Phi(q) - \log \Phi(p) = -\infty,
\]
it is not hard to see that by choosing the points $\{\lambda(t)\}_{t = 0}^{\infty}$ properly the right side is $-\infty$. Thus, along this sequence the relative value process $V(t)$ tends to zero as $t \rightarrow \infty$. This shows that $\pi$ cannot be a pseudo-arbitrage if zero is a limit point of $\Phi(K)$. This completes the proof of Theorem \ref{thm:charac1}.
\end{proof}

\subsection{The differentiable case}
In this subsection we derive differential inequalities satisfied by $C^1$ MCM portfolios. Let $\pi$ be a $C^1$ portfolio satisfying the MCM property. By Proposition \ref{lem:fgpproperties}, its generating function $\Phi$ is $C^2$. Recall that $T\left( q \mid p \right)$ is the L-divergence functional defined by \eqref{eq:discreteenergy}.

\begin{defn} [Drift quadratic form]
Let $\pi$ be generated by a $C^2$ positive concave function $\Phi$ on $\Delta^{(n)}$. The drift quadratic form of $(\pi, \Phi)$ is the quadratic form $H$ satisfying
\[
H(p)(v, v) := \frac{-1}{2\Phi(p)} \Hess \Phi(p)(v, v),
\]
where $p \in \Delta^{(n)}$ and $v \in T\Delta^{(n)}$. Here $\Hess \Phi$ is the Hessian of $\Phi$ regarded as a quadratic form. By definition, it is given by
\begin{equation}  \label{eqn:Hessian}
\Hess \Phi(p)(v, v) = \left.\frac{d^2}{dt^2} \Phi(p + tv)  \right|_{t = 0}.
\end{equation}
\end{defn}
 
Direct differentiation shows that $H$ is the Taylor series approximation of $T$. For $t \in {\Bbb R}$ small, we have
\begin{equation} \label{eq:Taylor}
T\left(p + tv \mid p \right) = H(p)(tv, tv) + o\left(t^2\right).
\end{equation}

Of special importance is the case where $\pi \in \overline{\Delta^{(n)}}$ is a constant-weighted portfolio. The corresponding drift quadratic form is called the {\it excess growth quadratic form}.

\begin{defn}[Excess growth]
Let $\pi \in \overline{\Delta^{(n)}}$. The excess growth quadratic form $\Gamma_{\pi}$ of $\pi$ is defined by
\begin{equation} \label{eqn:freeenergy}
\Gamma_{\pi}(p)(v, v) = \frac{1}{2} \sum_{i, j = 1}^n \frac{\pi_i(\delta_{ij} - \pi_j)}{p_ip_j}v_iv_j,
\end{equation}
where $p \in \Delta^{(n)}$ and $v \in T\Delta^{(n)}$.
\end{defn}

Finally we need the concept of {\it Fisher information metric} from information geometry (\cite[Section 2.5]{A07}).

\begin{defn}[Fisher information metric]
The Fisher information metric on $\Delta^{(n)}$ defines an inner product $\langle \! \langle \cdot, \cdot \rangle \! \rangle_p$ of tangent vectors for each $p \in \Delta^{(n)}$. If $u = (u_1, \ldots, u_n)$ and $v = (v_1, \ldots, v_n)$ are tangent vectors of $\Delta^{(n)}$, and $p \in \simp^{(n)}$, the inner product is defined by
\[
\langle \! \langle u, v \rangle \! \rangle_p = \frac{1}{2} \sum_{i = 1}^n \frac{1}{p_i}u_iv_i.
\]
\end{defn}

Notice that if $\pi = p$, then $\Gamma_{p}(p)(\cdot, \cdot) = \langle \! \langle \cdot, \cdot \rangle \! \rangle_{p}$, so the excess growth quadratic form of the market portfolio is the Fisher information metric. 

\begin{thm} \label{thm:concaveequivalence}
Let $\pi$ be a portfolio generated by a $C^2$ positive concave function on $\Delta^{(n)}$. Then:
\begin{enumerate}
\item[(i)] The weight ratio $w(\mu) = \frac{\pi}{\mu}$ satisfies 
\begin{equation} \label{eq:curvature}
\iprod{v, D_vw(p)} \le -\iprod{w(p), v}^2,
\end{equation}
for any $\mu \in \Delta^{(n)}$ and $v \in T\Delta^{(n)}$.
\item[(ii)] For any $p \in \Delta^{(n)}$ and $v \in T\Delta^{(n)}$, we have
\begin{equation} \label{eqn:energyinformation}
\Gamma_{\pi}(p)(v, v) - \langle \! \langle D_v \pi(p), v \rangle \! \rangle_{p} \geq 0.
\end{equation}
\end{enumerate}
\end{thm}

Recall that $D_vw$ and $D_v\pi$ are the push forwards of tangent vectors (images of the differential map). Intuitively, the inner product $\langle \! \langle D_v \pi(\mu), v \rangle \! \rangle_{\mu}$ measures how much the portfolio weights move in the direction of the increment of market weights. By \eqref{eqn:energyinformation}, in order that $\pi$ is a pseudo-arbitrage, the inner product cannot be more than the excess growth rate of the portfolio. This gives the meaning of concavity at the portfolio map level.

\begin{figure}[t!]
\centering
\includegraphics[scale=0.45]{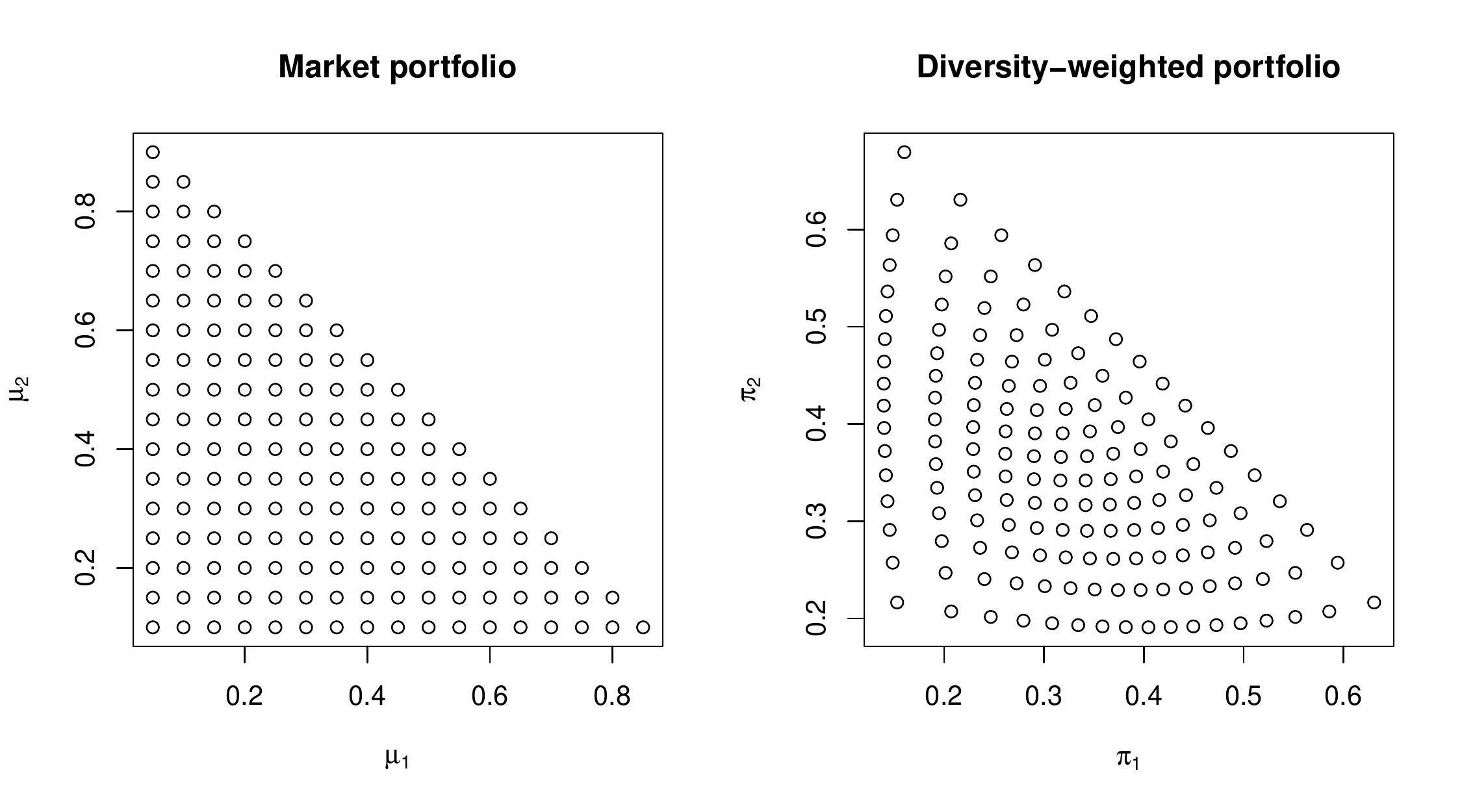}
\vspace{-15pt}
\caption{The portfolio map (projected to the first two coordinates) of the diversity-weighted portfolio where $\pi_i = \frac{\sqrt{\mu_i}}{\sum_{j = 1}^n \sqrt{\mu_j}}$ and $\Phi(\mu) = \left(\sum_{j = 1}^n \sqrt{\mu_j} \right)^2$, for $n=3$. Each point on the right is the image of a point on the left.}
\label{fig:backtest}
\end{figure}

\begin{proof}
(i) Let $p \in \Delta^{(n)}$, $v \in T\Delta^{(n)}$ and $\varepsilon > 0$. By the MCM property with $m = 1$, we get
\[
\left( 1 + \left\langle \frac{\pi(p)}{p}, \varepsilon v \right\rangle\right) \left( 1 - \left\langle \frac{\pi(p + \varepsilon v)}{p + \varepsilon v}, \varepsilon v \right\rangle\right) - 1 \geq 0.
\]
Expanding the above expression, we get
\[
-\varepsilon \langle w(p + \varepsilon v) - w(p), v \rangle - \varepsilon^2 \langle w(p), v \rangle \langle w(p + \varepsilon, v \rangle \geq 0.
\]
Dividing by $\varepsilon^2$ and taking the limit as $\varepsilon$ tends to zero, we get the desired inequality.

\medskip
\noindent
(ii) Without loss of generality, we may assume that $\Phi$ is twice continuously differentiable on an open neighborhood of $\Delta^{(n)}$ in ${\Bbb R}^n$. Using the product rule, we have
\begin{equation*}
\begin{split}
\frac{1}{2} \langle v, D_vw \rangle &= \frac{1}{2} \sum_{i, j = 1}^n v_iv_j \frac{\partial}{\partial \mu_j} \frac{\pi_i}{\mu_i} \\
   &= \frac{1}{2} \sum_{i, j = 1}^n v_iv_j \left( \frac{1}{\mu_i} \frac{\partial \pi_i}{\partial \mu_j} - \pi_i \delta_{ij} \frac{1}{\mu_j^2} \right) \\
   &= \langle \! \langle v, D_v \pi \rangle \! \rangle_{\mu} - \frac{1}{2} \sum_{i = 1}^n \frac{\pi_i}{\mu_i^2} v_i^2.
\end{split}
\end{equation*}
It follows that 
\begin{equation*}
\begin{split}
\frac{1}{2} \langle v, D_vw \rangle + \frac{1}{2} \langle w, v \rangle^2 &= \langle \! \langle v, D_v \pi \rangle \! \rangle_{\mu} - \frac{1}{2} \sum_{i, j = 1}^n \frac{\pi_i(\delta_{ij} - \pi_j)}{\mu_i\mu_j} v_iv_j \\
  &= \langle \! \langle v, D_v \pi \rangle \! \rangle_{\mu} - \Gamma_{\pi}(v, v).
\end{split}
\end{equation*}
\end{proof}

\begin{rmk}
In \cite[Section 4]{PW13}, we showed that when $n = 2$, a functionally generated portfolio
$\pi = (q(Y), 1 - q(Y))$, where $Y = \log \frac{\mu_1}{\mu_2}$ and $q: {\Bbb R} \rightarrow (0, 1)$, has a non-decreasing
drift process if and only if
\begin{equation} \label{eqn:twostockinequality}
q'(y) \leq q(y)(1 - q(y)), \quad y \in {\Bbb R}.
\end{equation}
Using the identity $d\langle Y \rangle = \frac{1}{(\mu_1\mu_2)^2}d\langle \mu_1 \rangle$, which
follows from It\^{o}'s formula, we can show that $\Gamma_{\pi}(d\mu, d\mu) = \frac{1}{2}q(1 - q)d\langle Y \rangle$ and
$\langle \! \langle d\mu, D_{d\mu} \pi(\mu) \rangle \! \rangle_{\mu} = \frac{1}{2}q'd\langle Y \rangle$.
We get
\[
d\Theta = \Gamma_{\pi}(d\mu, d\mu) - \langle \! \langle d\mu, D_{d\mu} \pi(\mu) \rangle \! \rangle_{\mu} = \frac{1}{2}\left(q(1 - q) - q'\right)d\langle Y \rangle.
\]
Thus Theorem \ref{thm:concaveequivalence} generalizes \eqref{eqn:twostockinequality} and provides a geometric interpretation.
\end{rmk}

\begin{rmk}
By Proposition \ref{lem:superdiff} and Theorem \ref{thm:concaveequivalence}, for a continuously differentiable portfolio $\pi$ the MCM property implies that the vector field $w = \pi / \mu$ is conservative (in the sense of \eqref{eq:consvec} where $\gamma$ is any curve in $\Delta^{(n)}$) and satisfies the inequality \eqref{eq:curvature}. The converse is also true and here is a sketch of proof. If $\pi / \mu$ is conservative, the line integral $\int_{\gamma} \frac{\pi}{\mu} d\mu$ defines via \eqref{eqn:potential} a function $\Phi$ on $\Delta^{(n)}$ such that $\pi$ is given by \eqref{eq:fgeqn}. Moreover, \eqref{eq:curvature} implies that $\Phi$ is concave. Thus $\pi$ is generated by a concave function and is MCM.

When $n = 2$, all continuously differentiable vector fields on $\Delta^{(2)}$ are conservative (also see \cite[Lemma 4.6]{PW13}) and thus \eqref{eq:curvature} implies the MCM property. In general, a portfolio may satisfy \eqref{eq:curvature} without being generated by a concave function. The following example is inspired by \cite[Section 24, page 240]{R70}. The construction works for any $n \geq 3$, but for concreteness we let $n = 3$. Consider the matrices
\[
A = \left( \begin{array}{ccc}
-1 & -1 & -1 \\
0 & -1 & -1 \\
0 & 0 & -1 \end{array} \right), \quad B = \frac{1}{2}(A' + A),
\]
where $A'$ is the transpose of $A$. Let $\lambda > 0$ be a parameter to be chosen, and define $A_{\lambda} := \lambda A$, $B_{\lambda} := \lambda B$. Note that $A$ is non-symmetric and $B$ is (strictly) negative definite. Here we use matrix notation whenever convenient. The matrix $A_{\lambda}$ defines a portfolio $\pi$ via the weight ratio, given by
\begin{equation} \label{eq:counterexample}
\begin{split}
w(\mu) &= \frac{\pi}{\mu} = A_{\lambda}\mu + \alpha_{\lambda}(\mu) {\bf 1} \\
       &= \lambda \left( \begin{array}{c} -1 \\ -(\mu_2 + \mu_3) \\ -\mu_3 \end{array} \right) + \alpha_{\lambda}(\mu) \left( \begin{array}{c} 1 \\ 1 \\ 1 \end{array} \right),
\end{split}
\end{equation}
where $\alpha_{\lambda}(\cdot): \Delta^{(n)} \rightarrow {\Bbb R}$ is some smooth function. Using the identity $\sum_i \pi_i = \sum_i \mu_i w_i = 1$, we have
\begin{equation} \label{eq:examplealpha}
\alpha_{\lambda}(\mu) = 1 + \lambda \left[\mu_1 + \mu_2(\mu_2 + \mu_3) + \mu_3^2\right] > 1.
\end{equation}
The portfolio is then given by
\[
\pi(\mu) = \mu \cdot A_{\lambda} \mu + \alpha(\mu) \mu.
\]
When $\lambda < 1$, the entries of $A_{\lambda}\mu$ are greater than $-1$. It follows from \eqref{eq:examplealpha} that $w > 0$, so the portfolio has positive weights.

If this portfolio is generated by a concave function, Proposition \ref{lem:fgpproperties} implies that the generating function is smooth. By \cite[Proposition 3.1.11]{F02}, there exists a continuously differentiable function $F$ on a neighborhood of $\Delta^{(n)}$ in ${\Bbb R}^n$ such that $\sum_i (w_i + F) d\mu_i$ is an exact differential $1$-form. It follows that $\frac{\partial}{\partial \mu_j} (w_i + F) = \frac{\partial}{\partial \mu_i} (w_j + F)$ for any $i$ and $j$. Letting $\widetilde{F} = \alpha + F$, we see by differentiating \eqref{eq:counterexample} that $\pi$
\begin{equation*}
\begin{split}
\frac{\partial}{\partial \mu_2}\left(w_1 + F\right) = \frac{\partial}{\partial \mu_2} \widetilde{F} =  \frac{\partial}{\partial \mu_1}\left(w_2 + F\right)= \frac{\partial}{\partial \mu_1} \widetilde{F}, \\
\frac{\partial}{\partial \mu_3}\left(w_1 + F\right) =\frac{\partial}{\partial \mu_3} \widetilde{F} = \frac{\partial}{\partial \mu_1}\left(w_3 + F\right)= \frac{\partial}{\partial \mu_1} \widetilde{F}, \\
\frac{\partial}{\partial \mu_3}\left(w_2 + F\right) = \lambda + \frac{\partial}{\partial \mu_3} \widetilde{F} = \frac{\partial}{\partial \mu_2}\left(w_3 + F\right)= \frac{\partial}{\partial \mu_2} \widetilde{F}.
\end{split}
\end{equation*}
So $\frac{\partial}{\partial \mu_1} \widetilde{F} = \frac{\partial}{\partial \mu_2} \widetilde{F} = \frac{\partial}{\partial \mu_3} \widetilde{F}$, while $\lambda + \frac{\partial}{\partial \mu_3} \widetilde{F} =\frac{\partial}{\partial \mu_2} \widetilde{F}$, which is clearly a contradiction. Thus the portfolio is not generated by a concave function.

It remains to check that \eqref{eq:curvature} holds. From \eqref{eq:counterexample}, we see that
\begin{equation} \label{eq:counterexample2}
\begin{split}
\langle v, D_vw \rangle + \langle v, w \rangle^2 &= v' B_{\lambda} v + (v' A_{\lambda} \mu)^2 \\
          &= \lambda \left( v' B v + \lambda \left( v' A \mu \right)^2 \right).
\end{split}
\end{equation}
For each $\mu$, $v \mapsto \left( v' A \mu \right)^2$ is a non-negative definite quadratic form in $v$. By continuity, there exists a constant $C$ such that $\left( v' A \mu \right)^2 \leq C \|v\|^2$ for all $\mu \in \Delta^{(n)}$. Since $B$ is negative definite, by choosing $\lambda > 0$ sufficiently small, we can make the sum non-positive definite. Hence \eqref{eq:curvature} holds on $\Delta^{(n)}$.
\end{rmk}

\section{Optimal transport} \label{sec:optimaltransport}
In this section we study functionally generated portfolios from the point of view of optimal transport. The general optimal transport problem has been given in the Introduction. We begin by recalling the notion of $c$-cyclical monotonicity which plays a crucial role in characterizing optimal solutions.

\subsection{$c$-cyclical monotonicity}
Consider the Monge-Kantorovich optimal transport problem with state spaces ${\mathcal{X}}$ and ${\mathcal{Y}}$ with cost function $c$.

\begin{defn}[$c$-cyclical monotonicity] \label{def:cm}
Let $A$ be a subset of ${\mathcal X} \times {\mathcal Y}$. We say that $A$ is $c$-cyclical monotone if for any $m \geq 1$ and any sequence $\{(x(k), y(k))\}_{k = 1}^m$ in $A$, we have
\begin{equation} \label{eqn:CM}
\sum_{k = 1}^m c(x(k), y(k)) \leq \sum_{k = 1}^m c(x(k), y(k+1)),
\end{equation}
with the convention $x(m + 1) = x(1)$ and $y(m + 1) = y(1)$.
\end{defn}

It can be shown that $c$-cyclical monotonicity is equivalent to the assertion that 
\begin{equation} \label{eqn:CM}
\sum_{k = 1}^m c(x(k), y(k)) \leq \sum_{k = 1}^m c(x(k), y(\sigma(k)))
\end{equation}
where $\sigma$ is any permutation of $[m] = \{1, \ldots, m\}$. The following is a basic result in optimal transport theory.

\begin{thm} \label{thm:OT} \cite[Theorem 5.10, part (ii)]{V08}
Suppose that the cost function $c: {\mathcal X} \times {\mathcal Y} \rightarrow {\Bbb R}$ is continuous and bounded below. Let ${\mathcal P}$ and ${\mathcal Q}$ be (Borel) probability measures on ${\mathcal X}$ and ${\mathcal Y}$ respectively. Assume that the value of the transport problem \eqref{eqn:Kantorovich} is finite. Then ${\mathcal R} \in \Pi({\mathcal P}, {\mathcal Q})$ solves the optimal transport problem \eqref{eqn:Kantorovich} if and only if the support of ${\mathcal R}$ is $c$-cyclically monotone.
\end{thm}

\subsection{MCM and $c$-cylical monotonicity}
We consider the cost function $c$ on $\overline{\Delta^{(n)}} \times [-\infty, \infty)^n$ defined by
\begin{equation} \label{eqn:logpartition}
c(p, h) := \log \left( \sum_{i = 1}^n e^{h_i} p_i \right).
\end{equation}
The cost function is clearly continuous. Let $h: \Delta^{(n)} \rightarrow [-\infty, \infty)^n \setminus \{(-\infty, \ldots, -\infty)\}$ be a function (thought of as a transport map). The map defines a portfolio function $\pi: \Delta^{(n)} \rightarrow \overline{\Delta^{(n)}}$ via \eqref{eqn:changemeasure}:
\begin{equation} \label{eqn:changemeasure2}
\frac{\pi_i(\mu)}{\mu_i}=\frac{e^{h_i(\mu)}}{\E_{\mu}\left( \exp(h(\mu)) \right)}, \quad i=1,2,\ldots,n.
\end{equation}
The following proposition gives the link between the cost function \eqref{eqn:logpartition} and the MCM property.

\begin{prop} \label{lem:CMMCM}
Let $h: \Delta^{(n)} \rightarrow [-\infty, \infty)^n \setminus \{(-\infty, \ldots, -\infty)\}$. The graph of $h$ is $c$-cyclical monotone if and only if the portfolio $\pi$ defined by \eqref{eqn:changemeasure2} has the MCM property. 
\end{prop}
\begin{proof}
Let $\{\mu(t)\}_{t = 0}^{m + 1}$ be a market cycle. By \eqref{eqn:changemeasure2}, we have
\begin{equation*}
\begin{split}
\frac{V(t + 1)}{V(t)} &= \sum_{i = 1}^n \frac{\pi_i(\mu(t))\mu_i(t + 1)}{\mu_i(t)} \\
&= \frac{1}{\sum_{j = 1}^n \mu_j(t) e^{h_j(\mu(t))}} \sum_{i = 1}^n \mu_i(t + 1) e^{h_j(\mu(t))} = \frac{ \langle \mu(t + 1), e^{h(\mu(t))}\rangle}{\langle \mu(t), e^{h(\mu(t))} \rangle}.
\end{split}
\end{equation*}
Multiplying the above terms along the market cycle and taking log on both sides, we have
\begin{equation} \label{eqn:MCMconnection}
\log V(m+1) = \sum_{t = 0}^m c(\mu(t + 1), h(\mu(t))) - \sum_{t = 0}^m c(\mu(t), h(\mu(t))).
\end{equation}
Suppose $\pi$ is MCM. Then in \eqref{eqn:MCMconnection} we have $\log V(m + 1) \geq 0$. So
\[
\sum_{t = 0}^m c(\mu(t), h(t)) \leq \sum_{t = 0}^m c(\mu(t + 1), h(\mu(t))).
\]
Since the market cycle is arbitrary, this implies that the graph of $h$ is $c$-cyclical monotone. The same reasoning shows that $\pi$ is MCM if the graph is $c$-cyclical monotone.
\end{proof}

\begin{proof} [Proof of Theorem \ref{thm:charac2}]
To show the first statement, suppose we can prove that the support of any optimal coupling is $c$-cyclical monotone. Then the portfolio $\pi$ defined by \eqref{eqn:changemeasure} is MCM by Proposition \ref{lem:CMMCM}. Thus $\pi$ is generated by a concave function by Proposition \ref{thm:MCM}, and the rest follows from Theorem \ref{thm:charac1}. It remains to prove that the support is $c$-cyclical monotone.

Since the cost function \eqref{eqn:logpartition} is unbounded below, we will use a little trick so that Theorem \ref{thm:OT} can be applied. Let ${\mathcal R}$ be an optimal coupling with finite cost. Consider the sets
\[
Z_m = \left\{(\mu, h): \max_{1 \leq i \leq n} h_i \geq -m, \min_{1 \leq i \leq n} \mu_i \geq \frac{1}{m}\right\}, \quad m \geq 1.
\]
Also let ${\mathcal R}_m$ be the restriction of ${\mathcal R}$ to $Z_m$, normalized to have mass $1$. Since $Z_m$ increases to $\Delta_n \times [-\infty, \infty) \setminus \{(-\infty, \ldots, -\infty)\}$, by the continuity of measure ${\mathcal R}_m$ is well defined for $m$ sufficiently large. Let ${\mathcal P}_m$ and ${\mathcal Q}_m$ be the marginals of ${\mathcal R}_m$. Consider the optimal transport problem
\begin{equation} \label{eqn:OTrestriction}
\inf_{{\mathcal R}' \in \Pi({\mathcal P}_m, {\mathcal Q}_m)} \E \left[ c(p, h) \right]
\end{equation}
where $(p, h) \sim {\mathcal R}'$. By the well known restriction property of optimal transport (see \cite[Theorem 4.6]{V08} whose proof does not rely on any topological assumptions), ${\mathcal R}_m$ is optimal for \eqref{eqn:OTrestriction}. Now the cost function, when restricted to $Z_m$, is continuous and bounded below. So, by Theorem \ref{thm:OT} we see that $\text{supp}({\mathcal R}) \cap Z_m$ is $c$-cyclical monotone. Since $m$ is arbitrary, $\text{supp}({\mathcal R})$ is $c$-cyclical monotone. 

Now we argue the converse. Suppose the portfolio $\pi$ is a pseudo-arbitrage on $K \subset \Delta^{(n)}$. Let ${\mathcal P}$ be any probability measure on $K$, and let ${\mathcal Q}$ be the distribution of $h = F(\mu)$, where $\mu \in {\mathcal P}$. We claim that the coupling $(\mu, h)$ is optimal for the optimal transport problem $\inf \mathrm{E}\left[ c(\mu, h) \right]$. Since $h_i = \log \left( \pi_i(p) / p_i \right)$ is bounded below, ${\mathcal Q}$ is supported on $L_m = [-m, \infty)^n$ for some $m$. On the set $\overline{\Delta^{(n)}} \times L_m$, the cost function $c$ is continuous and bounded below. Thus the coupling $(\mu, F(\mu))$ is optimal by Theorem \ref{thm:OT}.
\end{proof} 

\subsection{MCM in exponential coordinates}
Now we restrict to portfolio functions with strictly positive weights and give an alternative formulation of the transport problem using the exponential coordinates. 

Recall that $\iota: \Delta^{(n)} \rightarrow {\Bbb R}^{n-1}$ defined by \eqref{eqn:definetheta} is the global coordinate system which gives the exponential coordinates. Given an arbitrary function $\phi: {\Bbb R}^{n - 1} \rightarrow {\Bbb R}^{n - 1}$, we can define a portfolio function $\pi: \Delta^{(n)} \rightarrow \overline{\Delta^{(n)}}$ by \eqref{eqn:portfexpcoord}. Conversely, given a portfolio map $\pi:\Delta^{(n)}\rightarrow \Delta^{(n)}$, one can define $\phi$ by setting $\iota(\pi(\mu)) = \theta - \phi(\theta)$ whenever $\iota(\mu) = \theta$. Rearranging gives
\eq\label{eq:thetap}
\phi_i(\theta) =  \log\left( \frac{\mu_i}{\mu_n}  \right) - \log\left( \frac{\pi_i(\mu)}{\pi_n(\mu)}  \right),\quad i=1,2,\ldots, n-1.
\en

Fix a portfolio function $\pi$. Now we express the evolution of $V(t)$ in terms of the exponential coordinates. Given the market weight sequence $\{\mu(t)\}_{t = 0}^{\infty}$, we define the exponential coordinate process by $\theta(t) = \iota(\mu(t))$. Using \eqref{eqn:thetatop}, we have
\begin{equation*}
\begin{split}
\frac{V(t + 1)}{V(t)} &= \sum_{i = 1}^n \pi_i(t) \frac{\mu_i(t + 1)}{\mu_i(t)} \\
  &= e^{-\psi(\theta(t) - \phi(\theta(t)))} \frac{e^{-\psi(\theta(t + 1))}}{e^{-\psi(\theta(t))}} \left(1 + \sum_{i = 1}^{n-1} e^{\theta_i(t) - \phi_i(\theta(t))} \frac{e^{\theta_i(t + 1)}}{e^{\theta_i(t)}} \right) \\
  &= e^{-\psi(\theta(t) - \phi(\theta(t))) + \psi(\theta(t)) - \psi(\theta(t + 1))} e^{\psi(\theta(t + 1) - \phi(\theta(t)))}.
\end{split}
\end{equation*}
Taking log on both sides, we have
\[
\log \frac{V(t + 1)}{V(t)} = \left[\psi(\theta(t)) - \psi(\theta(t + 1))\right] + \left[\psi(\theta(t + 1) - \phi(t)) - \psi(\theta(t) - \phi(\theta(t)))\right]
\]
Summing over time, we get 
\begin{equation} \label{eqn:expdecomp}
\log V(t) = \psi(\theta(0)) - \psi(\theta(t)) + \sum_{s = 0}^{t - 1} \left[\psi(\theta(s + 1) - \phi(s)) - \psi(\theta(s) - \phi(\theta(s)))\right].
\end{equation}

Now consider a discrete cycle $\{\mu(t)\}_{t = 0}^{m + 1}$ where $\mu(m + 1) = \mu(0)$. Putting $t = m + 1$ in \eqref{eqn:expdecomp}, we have
\[
\log V(m + 1) = \sum_{t = 0}^{m} \left[\psi(\theta(t + 1) - \phi(t)) - \psi(\theta(t) - \phi(\theta(s)))\right].
\] 
Thus, if $\pi$ is MCM, then
\begin{equation} \label{eqn:expMCM}
\sum_{t = 0}^{m}  \psi(\theta(t) - \phi(\theta(t))) \leq \sum_{t = 0}^{m} \psi(\theta(t + 1) - \phi(\theta(s)))
\end{equation}
which is the definition of $c$-cyclical monotonicity. We summarize the above discussion by the following proposition. It allows us to pose the transport problem in terms of the exponential coordinates.

\begin{prop} \label{prop:MCMCM}
Let $c$ be the cost function on ${\Bbb R}^{n-1} \times {\Bbb R}^{n - 1}$ defined by
\begin{equation} \label{eqn:convexcost}
c(\theta, \phi) = \psi(\theta - \phi).
\end{equation}
Let $\pi$ be a portfolio and let $\phi: {\Bbb R}^{n - 1} \rightarrow {\Bbb R}^{n - 1}$ be the corresponding negative shift in exponential coordinates defined by \eqref{eqn:portfexpcoord} or \eqref{eq:thetap}. Then $\pi$ satisfies MCM if and only if the graph of $\phi$ is $c$-cyclical monotone.
\end{prop}

\subsection{Relative entropy as a cost function}
In this subsection we give another optimal transport problem giving rise to functionally generated portfolios. Now we take ${\mathcal{X}} = {\mathcal{Y}} = \overline{\Delta^{(n)}}$ and consider the cost function $\widetilde{c}$ on $\overline{\Delta^{(n)}} \times \overline{\Delta^{(n)}}$ given by the negative relative entropy:
\begin{equation} \label{eqn:newcost}
\widetilde{c}(p, q) = -H(q | p) = -\sum_{i = 1}^n q_i \log \frac{q_i}{p_i}.
\end{equation}
Here $p$ is interpreted as the market weight and $q$ is interpreted as the portfolio weight. Note that $\tilde{c}$ is non-positive and is $0$ if and only if $p = q$.

\begin{prop} \label{thm:reMCM}
Any $\widetilde{c}$-cyclical monotone subset of $\Delta^{(n)} \times \overline{\Delta^{(n)}}$ satisfies MCM. The converse is false. 
\end{prop}
\begin{proof}
Let $A$ be a $\widetilde{c}$-cyclical monotone subset of $\Delta^{(n)} \times \overline{\Delta^{(n)}}$, and let $\{(p(t), q(t))\}_{t = 0}^{m+1}$ be a cycle in $A$. Then $\widetilde{c}$-cyclical monotonicity implies that
\[
\sum_{t = 0}^m -H(q(t) \mid p(t)) \leq  \sum_{t = 0}^m -H(q(t) \mid p(t + 1)).
\]
Expanding, we have
\[
\sum_{t = 0}^m \sum_{i = 1}^n q_i(t) \log \frac{q_i(t)}{p_i(t)} \geq \sum_{t = 0}^m \sum_{i = 1}^n q_i(t) \log \frac{q_i(t)}{p_i(t + 1)}.
\]
Thus we get
\[
\sum_{t = 0}^m \left[ \sum_{i = 1}^n q_i(t) \log \frac{p_i(t + 1)}{p_i(t)} \right] \geq 0.
\]
Applying Jensen's inequality to the sum inside the square bracket for each $t$, we have
\[
\sum_{t = 0}^m \left[ \log \sum_{i = 1}^n q_i(t)\frac{p_i(t + 1)}{p_i(t)} \right]\geq 0
\]
and the MCM property follows by exponentiating.

To see that the converse is false, let $\pi$ be the identity map on $\Delta^{(n)}$. This is the market portfolio, which is MCM. But $\widetilde{c}(p, \pi(p)) \equiv 0$, so the market portfolio maximizes the transportation cost. Any other rearrangement will increase relative entropy and decrease the transportation cost.   
\end{proof}

Consider the transport problem with cost \eqref{eqn:newcost}. Since any optimal coupling have $\widetilde{c}$-cyclical monotone support (argue as in the proof of Theorem \ref{thm:charac2}), the support is MCM by Proposition \ref{thm:reMCM} and thus is generated by a concave function. We give an alternative treatment which relates the problem with one with {\it quadratic cost}.

\begin{thm}\label{lem:optimaltrans}
Let $U \subset \Delta^{(n)}$ and $V \subset \overline{\Delta^{(n)}}$. Let $\PP$ be supported on $U$ and let $\QQ$ be supported on $V$. Suppose that
\eq\label{eq:bndopt}
\gamma:=\sup_{\pi \in V, \; \mu \in U } H(\pi \mid \mu) < \infty.
\en
Then the value of the problem \eqref{eqn:Kantorovich} with cost function $\widetilde{c}$ is finite. Moreover,
\begin{enumerate}[(i)]
\item There exists a concave function $\Phi:\simp^{(n)}\rightarrow (0, \infty)$ such that any optimal coupling $\RR$ is concentrated on the set
\[
\left\{ (p, q)\in U \times \overline{\Delta^{(n)}}:\; \left(\frac{q_i}{p_i} - \frac{1}{n} \sum_{j = 1}^n \frac{q_j}{p_j} \right)_{1 \leq i \leq n} \in \partial \log \Phi(p)   \right\}.
\]
\item For any $p, q \in U$ we have $\Phi(q)/\Phi(p) \le \exp(\gamma)$. Thus $\log \Phi$ is bounded on $U$.
\end{enumerate}
\end{thm}

The proof of the above theorem depends on the following lemma. Here, if $p$ is a vector with positive components, $\log p$ means applying log to each coordinate.

\begin{lemma}\label{lem:logtonolog}
Let $\varphi$ be a proper convex function on $(0, \infty)^n$. Consider $-\log \simp^{(n)}$ as a subset of $(0,\infty)^n$. Assume that, for every point $p\in \simp^{(n)}$, there is a $\pi\in \overline{\simp^{(n)}}$ such that the following subdifferential inequality holds
\eq\label{eq:logtonolog}
\varphi\left( -\log q \right) \ge \varphi\left( - \log p\right) + \iprod{\pi, \log p - \log q}, \quad \text{for all $q\in \simp^{(n)}$}.
\en
That is, $\pi$ is a subdifferential of $\varphi$ at $-\log p$. Then $\Phi(p)=\exp\left[ - \varphi(-\log p) \right]$ is a concave function on $\simp^{(n)}$. Moreover, the portfolio function $\pi$ is generated by $\Phi$.
\end{lemma}

\begin{proof}
Consider any two points $p,q \in \simp^{(n)}$. By \eqref{eq:logtonolog} we get,
\[
\begin{split}
\varphi\left( - \log p\right) - \varphi\left( -\log q \right) & \le \iprod{\pi, \log q - \log p} = \sum_{i=1}^n \pi_i(p) \log\left( \frac{q_i}{p_i} \right)\\
&\le \log\left( \sum_{i=1}^n \pi_i \frac{q_i}{p_i}  \right), \quad \text{by Jensen's inequality},\\
&= \log \left( 1 +  \iprod{\frac{\pi(p)}{p}, q-p} \right).
\end{split}
\]

Exponentiating both sides of the above inequality, we get
\[
 \left( 1 +  \iprod{\frac{\pi}{p}, q-p} \right) \ge \exp\left[ \varphi\left( - \log p\right) - \varphi\left( -\log q \right)  \right] = \frac{\Phi(q)}{\Phi(p)}.
\]
Multiplying both sides by $\Phi(p)$, we see that $\Phi$ has a supergradient at $p$ given by the vector $\Phi(p)\left(\frac{\pi_i(p)}{p_i} - \frac{1}{n} \sum_{j = 1}^n \frac{\pi_j(p)}{p_j}\right)_{1 \leq i \leq n}$. Thus $\Phi$ is concave by \cite[Theorem 3.2.6]{BSS13}. The last statement follows from Proposition \ref{thm:MCM} and Proposition \ref{lem:superdiff}.
\end{proof}

\begin{proof}[Proof of Theorem \ref{lem:optimaltrans}] It is clear that the value of the problem is finite.

\medskip

To prove (i) we will use the Knott-Smith optimality criterion (\cite[Theorem 2.12]{V03}) which is a fundamental result for optimal transport with quadratic cost. First we will reformulate problem \eqref{eqn:Kantorovich}. Let $\widetilde\PP$ be the law of $\zeta=-\log \mu$ when $\mu \sim \PP$. Then both $\PPT$ and $\QQ$ are probability measures on $\rr^n$, and the optimization problem \eqref{eqn:Kantorovich} becomes $\sup_{\RR \in \Pi\left( Q, \PPT  \right)} \E_{\RR} \left(  \iprod{\pi, \zeta} - H(\pi)\right)$, where $H(\pi)$ is the Shannon entropy of $\pi$. Since $\E_{{\mathcal R}}[H(\pi)] = \E_{{\mathcal Q}}[H(\pi)]$, this term can be dropped. Completing the squares, we see that the transport problem is equivalent to
\eq\label{eq:otquad}
\inf_{\RR \in \Pi\left( Q, \PPT  \right)} \E_{\RR} \left(  \norm{\pi - \zeta}^2 \right),
\en
which is the usual optimal transport problem for the quadratic cost.

Let $\widetilde{\RR}$ be an optimal solution to \eqref{eq:otquad} (which exists by \cite[Theorem 4.1]{V08}). The Knott-Smith optimality criterion states that there is a lower semi-continuous convex function $\varphi$ on $\rr^n$ such that the support of the optimal coupling is contained in the graph of $\partial \varphi$. That is, for $\widetilde{\RR}$ almost all $(\pi, \zeta)$ we have $\pi \in  \partial \varphi(\zeta)$. Consider the map $\mu=\exp(-\zeta) \mapsto \pi$ and let $\varsigma(\mu)= -\varphi(-\log \mu)=-\varphi(\zeta)$. By Lemma \ref{lem:logtonolog}, $\Phi=\exp(\varsigma(\mu))$ is a concave function which generates $\pi$.

\medskip

For (ii), let $(q,\pi)$ be in the support of $\widetilde\RR$ as above. By convexity of $\varphi$ and the fact that $\pi$ is a supergradient at $-\log q$, we have
\[
\varphi\left(-\log p\right) - \varphi\left( -\log q  \right) \ge \iprod{\pi, - \log p + \log q} = H\left(\pi \mid p\right) - H\left( \pi \mid q\right) \ge -\gamma.
\]
Multiplying the above inequalities by the negative sign and then exponentiating prove our claim.
\end{proof}

\subsection{Examples}

Now we give some simple examples where optimality can be verified directly; a more realistic application will be given in Section \ref{sec:example}. Let $S_n$ be the set of permutations of $(1, \ldots, n)$. For $\sigma = \left( \sigma_1, \ldots, \sigma_n  \right) \in S_n$ and $x \in {\Bbb R}^n$, let $\sigma\cdot x$ denote the vector $\left( x_{\sigma_1}, \ldots, x_{\sigma_n} \right)$.
Let $E_\sigma$ be the set defined by
\eq\label{eq:whatisesigma}
E_\sigma= \left\{  x\in \rr^n:  x_{\sigma_1} > x_{\sigma_2} > \cdots > x_{\sigma_n}    \right\}.
\en
In other words, the rank of the $\sigma_i$th coordinate of a point in $E_\sigma$ is $i$. Let $\mathrm{Id} \in S_n$ be the identity. Note that $E_\sigma= \sigma\cdot E_{\mathrm{Id}}$, by extending the action to sets.

\begin{exm}\label{exmp:ot1}
In this example the cost function is \eqref{eqn:logpartition}. Let $U$ be a precompact subset of $\simp^{(n)}$ which is defined only by the ranked coordinates, i.e., $U \cap E_\sigma = \sigma \cdot \left( U \cap E_{\mathrm{Id}} \right)$ for all $\sigma$. Let $\PP$ be an probability measure on $U$ which is invariant under relabelings of the coordinates. That is, $\PP$ is completely specified by the distribution of the ordered statistics $(\mu_{(1)}, \ldots, \mu_{(n)})$ and we say that $\PP$ is {\it exchangeable}. This will imply that the resulting portfolio functions are {\it rank-based}. We also assume that $\PP$ does not charge the boundaries of $E_\sigma$'s.

Let $V=\left\{\log e(i), \; 1\le i \le n  \right\}$ (the $\log$ is applied componentwise) and let $\QQ$ be the uniform distribution on this finite set. Then the optimal coupling of $(\PP, \QQ)$ can be described as follows. First, for every $p \in U \cap E_{\mathrm{Id}}$, we couple $p$ with $h = \log e(n)$. Next we extend by symmetry: for every $p \in U \cap E_\sigma$, couple $p$ with $h = \log e({\sigma_n})$. It can be checked that the resulting portfolio invests everything in the smallest stock. Similar examples can be worked out for $V=\left\{  \log \left( e(i) + e(j) \right), \; 1\le i < j\le n  \right\}$. In this case the optimal portfolio will invest according to the rescaled market weights in the lowest two stocks. If $V$ is the set of the log of the sums of $k$ many distinct $e(i)'s$, the resulting portfolio invests according to the market in the smallest $k$ stocks. This is an example of portfolio selected by {\it rank}, see \cite[Example 4.3.2]{F02}.
\end{exm}

\begin{exm}
Now we consider the negative relative entropy cost $\eqref{eqn:newcost}$. Let $U$ and $\PP$ be as in Example \ref{exmp:ot1}. Let $V=\left\{  e(i), \; 1\le i \le n  \right\}$ and let $\QQ$ be the uniform distribution on $V$. It is then clear that the following coupling minimizes the cost: With every $p \in U \cap E_1$, couple $p$ with $e(n)$ and extend by symmetry. Now the optimal coupling is again the portfolio that puts its entire holding on the smallest stock. However, if we take $V=\left\{  (e(i) + e(j))/2, \; i< j  \right\}$ or the set of averages of $k$ many distinct $e(i)$'s, the optimal portfolio will invest {\it equally} among the smallest $k$ stocks.
\end{exm}


\section{Empirical examples: two stocks case} \label{sec:example}
In general, solving optimal transport problems (either analytically or numerically) is a difficult task; see for example \cite{BFO14} and the references therein. Designing practical algorithms for solving the transport problem with cost \eqref{eqn:logpartition} or \eqref{eqn:convexcost} is an interesting open problem. In the case $n = 2$, the solution can be characterized explicitly due to the special structure of the real line and the convexity of the cost function \eqref{eqn:convexcost}. In this section we present the solution and give several empirical examples.

\subsection{Monotone rearrangements}
Throughout this section we assume $n = 2$. A typical point $\mu$ in $\Delta^{(2)}$ is represented as
\[
\mu = \left( \frac{e^{\theta}}{1 + e^{\theta}}, \frac{1}{1 + e^{\theta}}\right),
\]
where $\theta \in {\Bbb R}$ is the exponential coordinate of $\mu$. A portfolio vector $\pi(\mu)$ with positive weights corresponding to $\mu$ can be expressed as
\[
\pi(\mu) = \left( \frac{e^{\theta - \phi}}{1 + e^{\theta - \phi}}, \frac{1}{1 + e^{\theta - \phi}}\right)
\]
for some $\phi \in {\Bbb R}$. So the exponential coordinate of $\pi(\mu)$ is $\theta - \phi$. We will choose $\phi$ as a function of $\theta$. As $\phi$ increases, the portfolio underweights more and more stock 1 relative to the market weight. See Figure \ref{fig:phi} for the dependence of the portfolio on $\phi$ at different points on the simplex. As an example, if $|\phi|$ is bounded by $0.6$, the graph of the resulting portfolio will lie within the curves labeled $-0.6$ and $0.6$.

\begin{figure}[t!]
\centering
\includegraphics[scale=0.45]{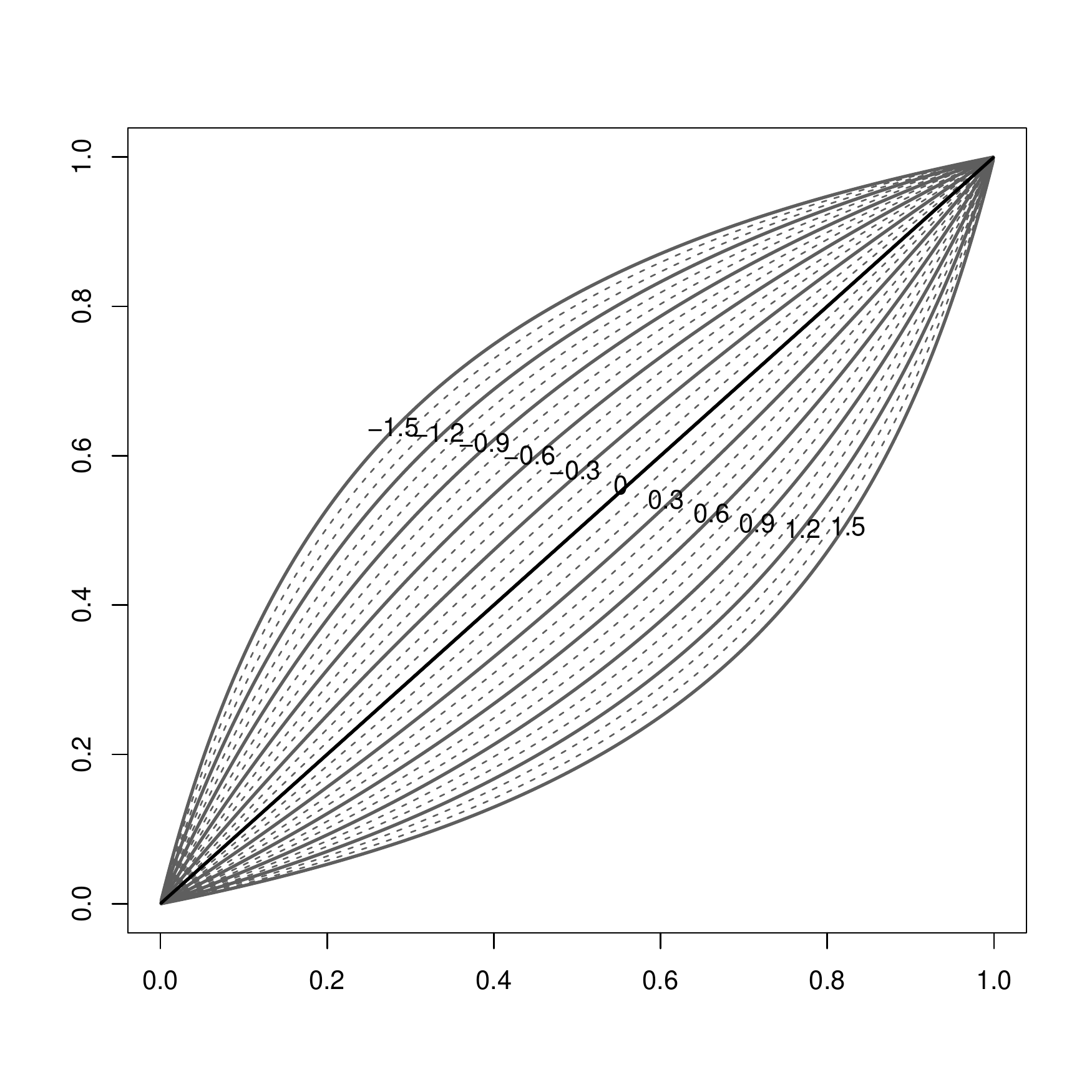}
\vspace{-15pt}
\caption{Plot of $\pi_1(\mu) =  \frac{e^{\theta - \phi}}{1 + e^{\theta - \phi}}$ as a function of $\mu_1 = \frac{e^{\theta}}{1 + e^{\theta}}$, for different values of $\phi$ (labeled).}
\label{fig:phi}
\end{figure}

Consider the transport problem with cost \eqref{eqn:convexcost}. Let $\widetilde{{\mathcal P}}$ and $\widetilde{{\mathcal Q}}$ be probability measures on ${\Bbb R}$. We assume that $\widetilde{{\mathcal P}}$ is absolutely continuous with respect to the Lebesgue measure. The cost is
\[
c(\theta, \phi) = \psi(\theta - \phi) = \log \left(1 + e^{\theta - \phi} \right), \quad \theta, \phi \in {\Bbb R}.
\]
Here
\[
\psi(x) = \log \left(1 + e^x\right)
\]
is a smooth and strictly convex function on ${\Bbb R}$. The transport problem is
\begin{equation} \label{eqn:twostockproblem}
{\Bbb E}_{\widetilde{{\mathcal{R}}} \in \Pi(\widetilde{{\mathcal P}}, \widetilde{{\mathcal Q}})} \psi(\theta - \phi)
\end{equation}
where $(\theta, \phi) \sim \widetilde{{\mathcal{R}}}$.

Let $G$ and $H$ be the distribution functions of $\widetilde{{\mathcal P}}$ and $\widetilde{{\mathcal Q}}$ respectively.

\begin{defn} [Monotone rearrangement]
The monotone transport map from $\widetilde{{\mathcal P}}$ to $\widetilde{{\mathcal Q}}$ is the map $F: {\Bbb R} \rightarrow {\Bbb R}$ defined by
\begin{equation} \label{eqn:monotonetransport}
F(x) = \inf\{y: H(y) \geq G(x)\}.
\end{equation}
\end{defn}

In other words, $F$ is defined by matching the quantiles of $H$ to those of $G$. It is clear from \eqref{eqn:monotonetransport} that $F$ is non-decreasing. Moreover, it is easy to check that if $\theta \sim \widetilde{{\mathcal P}}$, then $F(\theta) \sim \widetilde{{\mathcal Q}}$. Thus $(\theta, F(\theta))$ is a coupling of $(\widetilde{{\mathcal P}}, \widetilde{{\mathcal Q}})$. In fact, $F$ is the unique non-decreasing function (up to the null sets of $\widetilde{{\mathcal P}}$) which maps $\widetilde{{\mathcal P}}$ to $\widetilde{{\mathcal Q}}$. 

The following theorem is a special case of a well-known fact (see for example \cite[Theorem 3.1]{JS12}). Indeed, the monotone transport map remains optimal if $\psi$ is replaced by {\it any} strictly convex function. 

\begin{thm} \label{thm:monotone}
The coupling $(\theta, F(\theta))$ where $\theta \sim \widetilde{{\mathcal P}}$ and $F$ is the monotone transport map from $\widetilde{{\mathcal P}}$ to $\widetilde{{\mathcal Q}}$ is the unique solution to the transport problem \eqref{eqn:twostockproblem}.
\end{thm}

An explicit example is where $\widetilde{{\mathcal{P}}}$ and $\widetilde{{\mathcal{Q}}}$ are normal. In this case, the monotone transport map is linear and the corresponding portfolio is essentially a {\it diversity-weighted portfolio} (see \cite[Section 3.4]{F02}).

\begin{prop} \label{prop:twostockgaussian}
Let $\widetilde{{\mathcal{P}}} = N(m_1, \sigma_1^2)$ and $\widetilde{{\mathcal{Q}}} = N(m_2, \sigma_2^2)$. Then the monotone transport map is given by
\[
F(\theta) = m_2 + \frac{\sigma_2}{\sigma_1} (\theta - m_1), \quad \theta \in {\Bbb R}.
\]
Moreover, the portfolio function corresponding to the transport map $F$ is given by
\begin{equation} \label{eqn:weightedDW}
\pi(\mu) = \left(\frac{c\mu_1^{\alpha}}{c\mu_1^{\alpha} + \mu_2^{\alpha}}, \frac{\mu_2^{\alpha}}{c\mu_1^{\alpha} + \mu_2^{\alpha}}\right),
\end{equation}
where $\alpha = 1 - \frac{\sigma_2}{\sigma_1}$ and $c = \exp\left(\frac{\sigma_2}{\sigma_1} m_1 - m_2\right)$.
\end{prop}
\begin{proof}
Let $X \sim N(m_1, \sigma_1^2)$ and $Y \sim N(m_2, \sigma_2^2)$. The first statement follows from the fact that $(X - m_1) / \sigma_1$ and $(Y - m_2) / \sigma_2$ have the same distribution.

To show \eqref{eqn:weightedDW}, note that $\theta = \log \frac{\mu_1}{\mu_2}$ and, by definition of $\pi(\mu)$, we have
\begin{equation*}
\begin{split}
\log \frac{\pi_1(\mu)}{\pi_2(\mu)} = \theta - F(\theta) &= \left(1 - \frac{\sigma_2}{\sigma_1} \right) \log \frac{\mu_1}{\mu_2} + \left( \frac{\sigma_2}{\sigma_1} m_1 - m_2\right) \\
  &= \alpha \log \frac{\mu_1}{\mu_2} + \log c.
\end{split}
\end{equation*}
Rearranging gives the result.
\end{proof}

Clearly the portfolio has the form \eqref{eqn:weightedDW} whenever the transport map is linear, so the normality assumption is not required. Nevertheless, it is instructive to see how the portfolio depends on the means and variances of $\widetilde{{\mathcal{P}}}$ and $\widetilde{{\mathcal{Q}}}$. In particular, the exponent $\alpha$ in Proposition \ref{prop:twostockgaussian} depends on the ratio $\frac{\sigma_2}{\sigma_1}$. If $\sigma_1 = \sigma_2$, then $\alpha = 0$ and $\pi$ is a constant-weighted portfolio. If $0 < \sigma_2 < \sigma_1$, then $0 < \alpha < 1$ and $\pi$ is essentially the diversity-weighted portfolio. If $\sigma_2 > \sigma_1 > 0$, then $\alpha $ is negative and the corresponding portfolio is studied in the recent paper \cite{VK15}. On the other hand, the mean $m_2$ of $\widetilde{{\mathcal{Q}}}$ represents systematic overweight/underweight of stock 1 and interacts with other parameters to determine the constant $c$.

\subsection{Empirical examples}
\begin{figure}[t!]
\centering
\includegraphics[scale=0.45]{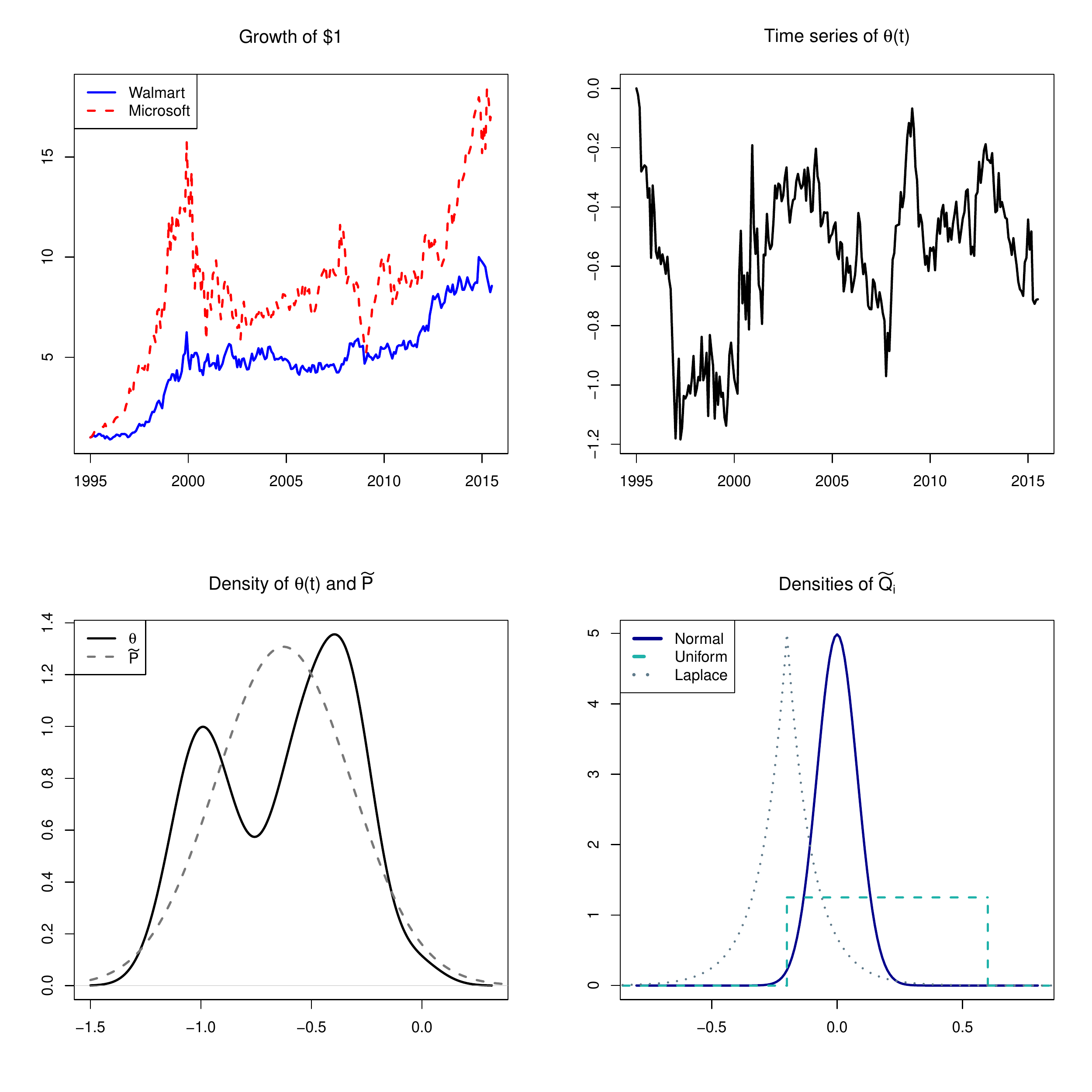}
\vspace{-15pt}
\caption{Plots of the data. Top left: Time series of the normalized stock prices. Top right: Time series of $\theta(t)$, the exponential coordinate process. Bottom left: Density estimate of $\theta(t)$ over the training period (solid curve) and density of $\widetilde{{\mathcal{P}}}$ (dashed curve). Bottom right: Densities of our choices of $\widetilde{{\mathcal{Q}}}$.}
\label{fig:data}
\end{figure}

In this subsection we use a simple example to illustrate how our methodology of optimal transport might be applied in practice. Consider the monthly stock prices of Walmart (stock 1) and Microsoft (stock 2) from January 1995 to July 2015. The stock prices (normalized to be $\$1$ at January 1995) are plotted in Figure \ref{fig:data} (top left). The `market' consists of the two stocks and the initial market weight is $(0.5, 0.5)$. We compute the exponential coordinate process $\theta(t) = \log \frac{\mu_1(t)}{\mu_2(t)}$ (top right). Suppose we use the first 10 years of data (120 months) as training data. Our objective is to use the training data as well as choices of ${\mathcal{Q}}$ to construct portfolios that will be backtested using the next 10 years of data. To do this using optimal transport, we need to specify the probability distributions $\widetilde{{\mathcal{P}}}$ and $\widetilde{{\mathcal{Q}}}$ on ${\Bbb R}$.

\medskip

{\it Choice of $\widetilde{{\mathcal{P}}}$.} The measure $\widetilde{{\mathcal{P}}}$ reflects our belief of the position of $\theta(t)$ in the future. Figure \ref{fig:data} plots the density estimate of $\theta(t)$ over the training period (bottom left). The distribution is bimodal (corresponding to the periods 1997-2000 and 2002-2004) and is mostly concentrated in the interval $[-1.2, 0]$. Suppose our belief is that the market weight will most likely remain in this region in the next decade. For simplicity, we take $\widetilde{{\mathcal{P}}}$ to be the normal distribution whose mean and standard deviation match those of the density estimate. Explicitly, we have
\[
\widetilde{{\mathcal{P}}} = N(-0.626, 0.305).
\]
A more diffuse distribution can be chosen if the investor is less certain.

\medskip

{\it Choice of $\widetilde{{\mathcal{Q}}}$.} Recall that the portfolio has the representation
\[
\pi(\mu) = \left( \frac{e^{\theta - \phi}}{1 + e^{\theta - \phi}}, \frac{1}{1 + e^{\theta - \phi}}\right),
\]
where $\phi$ is a function of $\theta$ and $\widetilde{{\mathcal{Q}}}$ is the marginal distribution of $\phi$, given that $\theta$ is distributed as $\widetilde{{\mathcal{P}}}$.

To illustrate the effects of different distributions we consider three distributions given as follow:
\begin{equation*}
\begin{split}
\widetilde{{\mathcal{Q}}}_1 &= N(0, 0.08), \\
\widetilde{{\mathcal{Q}}}_2 &= \mathrm{Uniform}(-0.2, 0.6),\\
\widetilde{{\mathcal{Q}}}_3 &= \mathrm{Laplace}(\mathrm{location} = -0.2 , \mathrm{scale} = 0.1).
\end{split}
\end{equation*}
Here we recall that the Laplace distribution with location parameter $a$ and scale parameter $b$ has density given by $f(x) = \frac{1}{b} \exp\left( -\frac{|x - a|}{b} \right)$. The densities of these distributions are shown in Figure \ref{fig:data} (bottom right). We denote the resulting portfolios by $\pi^{(1)}$, $\pi^{(2)}$ and $\pi^{(3)}$. 

Let us give some intuitions about these distributions. Overall, the distributions we choose concentrate in the interval $[-0.6, 0.6]$. From Figure \ref{fig:phi}, they allow moderate deviations from the market weight but not too much (most of the time).

Note that $\widetilde{{\mathcal{Q}}}_1$ has mean $0$ and has a rather small standard deviation (about a quarter of the standard deviation of $\widetilde{{\mathcal{P}}}$). This means that on average $\pi^{(1)}$ will not overweight or underweight stock 1 (Walmart) and the deviation is most of the time small. By Proposition \ref{prop:twostockgaussian}, we know that $\pi^{(1)}$ is a diversity-weighted portfolio with $\alpha = 1 - \frac{0.08}{0.305} \approx 0.74$. (From \eqref{eqn:weightedDW}, the portfolio is constant-weighted if $\alpha = 0$ and buy-and-hold if $\alpha = 1$.)

For $\widetilde{{\mathcal{Q}}}_2$, we expect that $\pi^{(2)}$ tends to underweight stock 1 (about $75\%$ of the time provided the future empirical distribution of $\theta(t)$ is close to $\widetilde{{\mathcal{P}}}$). Since $\widetilde{{\mathcal{Q}}}_2$ has bounded support, the weight ratios of $\pi^{(2)}$ are uniformly bounded on $\Delta^{(n)}$. However, the underweight can be significant on a certain region.

Finally, $\widetilde{{\mathcal{Q}}}_3$ has a Laplace distribution which has fatter tails than the normal distribution. Thus we expect that $\pi^{(3)}$ deviates more (from the market portfolio) than a diversity-weighted portfolio with matching parameters near the boundary of the simplex. Also $\widetilde{{\mathcal{Q}}}_3$ is chosen to have negative mean. Thus $\pi^{(3)}$ will tend to overweight stock 1. In practice, the location measure of $\widetilde{{\mathcal{Q}}}$ should reflect the investor's belief about the relative performances of the stocks in the future.

\begin{figure}[t!]
\centering
\includegraphics[scale=0.5]{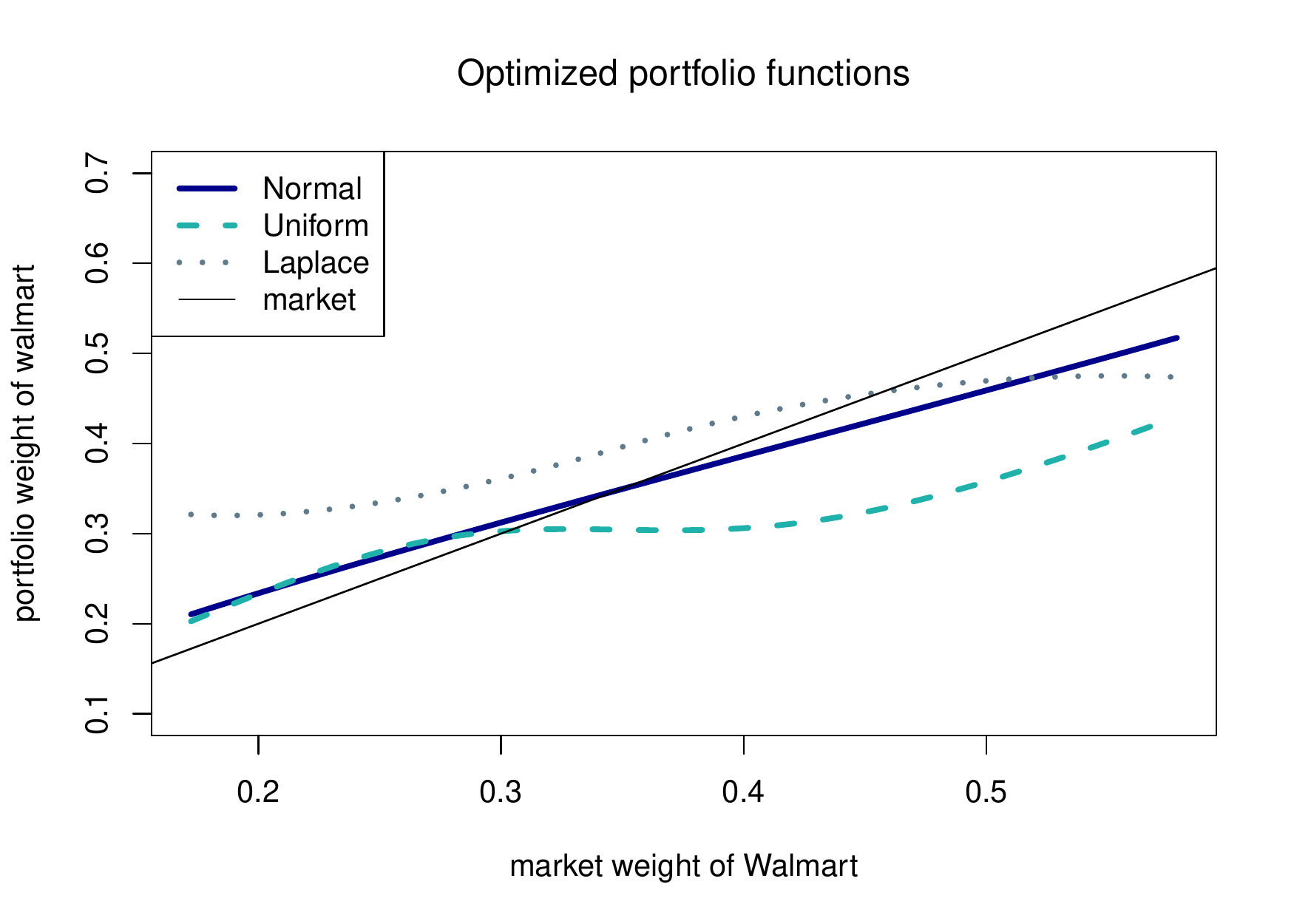}
\vspace{-15pt}
\caption{Plot of $\pi^{(i)}_1(\mu)$ against $\mu_1$, for $i = 1, 2, 3$. The curves are labeled using the distributions $\widetilde{{\mathcal{Q}}}_i$.}
\label{fig:result}
\end{figure}

\medskip

{\it Results.} For each choice of $\widetilde{{\mathcal{Q}}}$ we solve the optimal transport problem using the method of monotone rearrangement (Theorem \ref{thm:monotone}). The resulting portfolio functions $\pi^{(i)}$ are plotted in Figure \ref{fig:result}. The range of $\mu_1$ shown contains more than $99.9\%$ of the mass of $\widetilde{{\mathcal{P}}}$.

The features of the portfolios are consistent with our intuitions. As noted $\pi^{(1)}$ is a diversity-weighted portfolio which is quite close to the market portfolio by construction. Note that the curve intersects the market weight function around $\mu_1 = 0.35$. This corresponds to the median of $\widetilde{{\mathcal{P}}}$ and is a consequence of the fact that $\widetilde{{\mathcal{Q}}}_1$ is symmetric about $0$. Thus if $\widetilde{{\mathcal{P}}}$ is close to reality, $\pi^{(1)}$ will overweight stock 1 half of the time and underweight stock 1 half of the time.

The portfolio $\pi^{(2)}$ consistently underweights stock 1 because $\widetilde{{\mathcal{Q}}}_2$ is biased towards the right, and it has the largest deviation on the range shown. Nevertheless, if we draw the curves towards the boundary points $0$ and $1$, the boundedness of the support of $\widetilde{{\mathcal{Q}}}_2$ forces $\pi^{(2)}$ to be close to the market weight near the boundary of the simplex (in the sense that the weight ratios are bounded). This is not the case for $\pi^{(1)}$ and $\pi^{(3)}$ whose distributions have unbounded supports.

As for $\pi^{(3)}$, we note that most of the curve is above the market because $\widetilde{{\mathcal{Q}}}_3$ has negative mean. The portfolio deviates more and more towards the boundary because the Laplace distribution has fat tails. In this case, optimal transport couples large values of $|\phi|$ with the boundary values of $\mu$ which have small probability under $\widetilde{{\mathcal{P}}}$.

\medskip

{\it Backtesting.} Finally we compute the relative values of the three portfolios with respect to the market portfolio during the testing period 2005-2015. The result is shown in Figure \ref{fig:backtest}.

\begin{figure}[t!]
\centering
\includegraphics[scale=0.5]{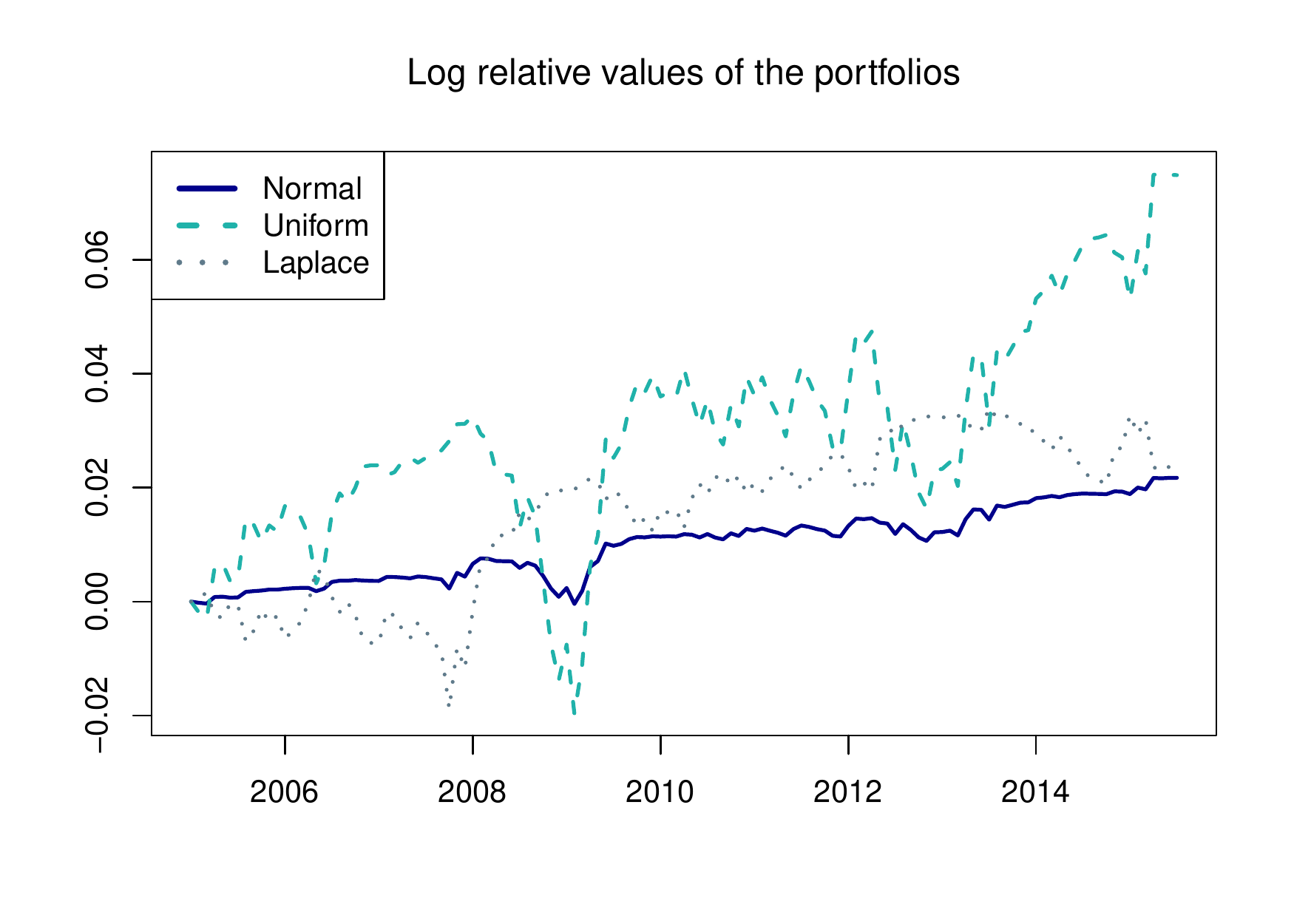}
\vspace{-15pt}
\caption{Log relative values of the portfolios $\pi^{(i)}$ in the testing period $2005-2015$.}
\label{fig:backtest}
\end{figure}

At the end of the period all three portfolios outperformed the market (by respectively $2.17\%$, $7.49\%$ and $2.38\%$, in log scale, over the 10 year period). The amounts are not large (except perhaps for $\pi^{(2)}$), and this is mostly because the portfolios deviate only moderately from the market portfolio. While detailed analysis of the performance is beyond the scope of the paper, we note that the relative riskiness of the portfolio (with respect to the market portfolio, also called the tracking error) depends on the deviation from the market weights and hence the location and dispersion of $\widetilde{{\mathcal{Q}}}$. The distribution $\widetilde{{\mathcal{Q}}}_2$ deviates most from $0$ and thus $\pi^{(2)}$ is riskier; it also has the biggest reward at the end of the period. Note that the approach of optimal transport optimizes a portfolio function over a region of $\Delta^{(n)}$ instead of picking portfolio weights period by period; it is simpler and perhaps more robust and prevents overfitting. See \cite[Section 5]{W14} for another approach of optimizing functionally generated portfolios in terms of the concavity of the generating function.

\bibliographystyle{amsalpha}
\bibliography{informgeo}

\end{document}